\pdfoutput=1

\documentclass[preprint,12pt]{elsarticle}
\usepackage{amsmath}
\usepackage{amsfonts}
\usepackage{amssymb}
\usepackage{graphicx}
\usepackage{algorithm}
\usepackage{algpseudocode}
\usepackage{mathrsfs}
\usepackage{caption}
\usepackage{subcaption}
\usepackage{pgffor}
\usepackage{amsthm}
\usepackage{placeins}
\usepackage{adjustbox}
\usepackage{color}
\usepackage{epstopdf}
\usepackage{appendix}

\newcommand{\norm}[1]{\left\lVert#1\right\rVert}
\newtheorem{theorem}{\indent\sc Theorem}

\newtheorem{proposition}{\indent\sc Proposition}

\DeclareMathOperator*{\argmin}{arg\,min}

\begin{document}

\bibliographystyle{unsrt}

\title{Bayesian and Variational Bayesian approaches for flows in heterogenous random media}
\author[fdu]{Keren Yang} 
\author[tamu_math,tamu_stat]{Nilabja Guha \corref{cor} }
\ead{nguha@math.tamu.edu}
\author[tamu_math]{Yalchin Efendiev} 
\author[tamu_stat]{Bani K. Mallick} 
\address[fdu]{School of Mathematical Sciences, Fudan University, Shanghai 200433, China}
\address[tamu_math]{Department of Mathematics, Texas A\&M University, College Station, TX 77843, USA}
\address[tamu_stat]{Department of Statistics, Texas A\&M University, College Station, TX 77843, USA}
\cortext[cor]{Corresponding author}


\begin{abstract}
In this paper, we study porous media flows in heterogeneous stochastic media.
We propose an efficient forward simulation technique that is tailored for
variational Bayesian inversion. As a starting point,
the proposed forward simulation
technique decomposes the solution into the sum of separable
functions (with respect to randomness and the space), where
each term is calculated based on a variational approach. This is similar to
  Proper Generalized Decomposition (PGD). Next, we apply a multiscale
technique to solve for each term (as in \cite{efendiev2013generalized})
and, further, decompose the random function
into 1D fields. As a result, our proposed method provides  an approximation
 hierarchy
for the solution as we increase the number of terms
in the expansion and, also, increase the spatial resolution of each term.
We use the hierarchical solution distributions in a variational Bayesian approximation
to perform uncertainty quantification in the inverse problem.
We conduct a detailed numerical study to explore the performance of the proposed uncertainty quantification technique and show the theoretical posterior concentration.
\end{abstract}
\maketitle

\section{Introduction}

In a physical system governed by differential equations, studying the uncertainty of the underlying system is of great interest.
Given the observation from the system  (possibly contaminated with errors)
inferencing  on the underlying parameter and its uncertainty constitutes
the uncertainty quantification of the inverse  problem \cite{kaipio2006statistical, kaipio2007statistical, tarantola2005inverse}. Bayesian methodology provides a natural framework for such problems, through specification of a prior distribution on the underlying parameter and the known likelihood function \cite{calvetti2007introduction, stuart2010inverse, mondal2010bayesian, xun2013parameter}. The Bayesian inference uses  Markov chain Monte Carlo (MCMC) or related methodologies. For MCMC, at  each sampling step,  for each proposed value of the parameter,  we need to solve the underlying equation i.e., the forward problem. The forward solution may not be closed form and we need to apply a numerical technique such as finite element or finite difference methods for that purpose, which can be computationally expensive.  {\it The objective of this paper} is to combine an efficient forward solution technique with an uncertainty quantification technique in an inverse problem. For that purpose, we use a separation of variable based model reduction technique in combination with a variational Bayesian approach for the inverse problem.

Our current development lies at the interface of
forward and inverse problems.  Estimating  subsurface properties
plays an
important role in many porous media applications, such as reservoir
characterization, groundwater modeling, {vadose} zone characterization, and so
on.  In this paper,
 we consider a model problem of
 reconstructing permeability field $\kappa(x)$
in the following single-phase flow equation
\begin{align}
  -\nabla\cdot(\kappa(x,\mu)\nabla u(x,\mu)) &= f \textrm{ in } \Omega_x, \nonumber \\
  u|_{\partial \Omega_x} &= 0,
\label{model}
\end{align}
{for given observations $y$ on pressure field $u$, where $\mu$ is the parameters associated with the permeability field $\kappa$, and $f$ is the known force term.}

Uncertainty quantification in an inverse problem can be a daunting task, particularly, for multiscale problems due to scale disparity. These problems require efficient forward simulation techniques that can reduce the degrees of freedom in a systematic way.

The first step in the estimation involves a parameterization
of the media properties using some prior information.
Typical approaches include
Karhunen-Lo\`{e}ve type expansion (KLE),
where the underlying random field is assumed to be Gaussian.
More complex parameterization
involves channelized fields \cite{mondal2014bayesian, iglesias2014well}.
These prior distributions can contain additional random parameters, which
require additional prior information and may involve techniques
like reversible jump MCMC \cite{mondal2014bayesian}.

For a forward solution technique used in a stochastic inversion,
we explore approaches that are based on separation of variables,
where we separate uncertainties and the spatial variables.
These separation
 approaches differ from Monte Carlo methods. They
 allow a fast computation of the solution space
over
the entire parameter range
 at a cost of computing each term
in the separable expansion.
In many cases, these separable expansions converge very fast and requires
only a few terms.
In our approaches, the solution is sought as a linear
combination of separable functions
(see e.g., \cite{ammar2006new, ammar2007new, le2009results} and applications to
inverse problems in
\cite{signorini2016proper, berger2016proper}).

In this paper, we consider an approach
proposed
in \cite{gao2016application}  with an overarching goal to do Bayesian inversion.
We show that these
approaches based on separation of variables
can be effectively used
within a variational Bayesian framework
and, in return, variational Bayesian approaches provide an effective
framework for applying separable solution approximation in inverse
problems.

Variational Bayesian techniques \cite{beal2003variational} offer an efficient method of posterior computation, where we approximate the posterior distribution by a closed form recursive estimate, without going into MCMC simulation.
 Variational Bayesian methods provide fast deterministic posterior computation and can be helpful for uncertainty quantification in inverse problem (cf., \cite{jin2010hierarchical, guha2015variational}).
In variational Bayesian methods, we assume that the  conditional posterior of different parts of the parameter are independent. To estimate each such part,   the full separability assumption {, i.e.}, the separability of the parametric and spatial part and separation of the parametric part over dimension is necessary. Therefore, the full separability gives us a natural framework for a variational algorithm.

For the prior distribution, we use KLE to parametrize the log permeability field
assuming Gaussian distribution. For coefficients representing
the separable solution form, we also use Gaussian process priors.

Further, we decompose the posterior solution into conditional independent parts and derive   a more efficient  variational solution. We show the posterior consistency of the proposed method that is the estimated value of the spatial field should be close to the true field under some appropriate metric if some  conditions hold.  We also show that the forward solution converges to the true solution.

We summarize the main contributions of this work as:
1) Proposing a fully separable representation for forward and inverse problem for the flow equation {\eqref{model}};
2) Proposing a fast UQ technique based on the separable representation;
3) Proposing a novel variational method, based on full separability and achieving further computational efficiency;
and
4) { Showing the convergence of both the forward solution and the posterior distribution under the proposed method.}

In the next section, we provide a detailed description of the model and data generating mechanism and a brief outline of the paper. First, we solve the forward problem under full separability and then, based on the forward solution,  we proceed to the inverse problem and uncertainty quantification.

\section{Outline of the paper}

We consider the flow equation \eqref{model}. If  $\kappa$ is given, solving $u$ is called a forward problem.
{ Given $u$, finding corresponding $\kappa$ constitutes the corresponding inverse problem. Given observed data $y$, we would like to estimate the underlying field $\kappa$. Here, $y$ is the quantity associated with \eqref{model} such as fractional flow or pressure data.} We further assume an error $\epsilon_f$ in the observation in the following form,
\begin{align}
y &= F(\kappa)+\epsilon_f, \nonumber \\
\epsilon_f& \sim {{\mathcal{N}}}(0,\sigma^2_f).
\label{ermodel}
\end{align}
Let $\pi(\kappa)$ be the prior on $\kappa$ and the posterior $\Pi$ is therefore
\begin{equation}
\Pi(\kappa|y) \propto \underbrace{p(y|\kappa)}_{\text{Likelihood}}\underbrace{\pi(\kappa)}_{\text{Prior}}.
\label{posterior}
\end{equation}
The likelihood computation part {in each iteration} involves a forward solve and therefore can be computationally expensive. The following separable representation can allow a fast likelihood computation.

Separating the spatial and parametric part,  we  can write, $u = \sum^{N_1}_{i=1} a_i(\mu)v_i(x)$, where $a_i(\mu)$ depends on the parametric part of $\kappa$, and $v_i(x)$ depends on the spatial variable (\cite{gao2016application}). If $\kappa$ is parameterized by finite dimensional $\mu=\{\theta_1,\dots,\theta_{N_2}\}$, we impose another separability condition $a_i(\mu)=\prod_{j=1}^{N_2}a_{i,j}(\theta_j)$.  With this complete separable representation, we can avoid estimating $a_i(\mu)$ in higher dimension and instead estimate low dimensional function $a_{i,j}$. Hence, we have  the following  fully separable approximation of $u$,
 \begin{equation}
    \tilde{u} = \sum^{N_1}_{i=1} \left(\prod^{N_2}_{j=1} a_{i,j}(\theta_j)\right)v_i(x).
\label{sep1}
\end{equation}

Under \eqref{sep1}, {the solution of $u$ can be expressed as a function depending on parameter $\mu\in\Omega_\mu$.} To solve the inverse problem and quantify the uncertainty, we explicitly solve the $a_{i,j}$'s over a grid of $\theta_j$'s and using a Gaussian process prior on the $a_{i,j}$s,  the posterior distribution of $\theta_j$'s can be computed. For  the parametrization of  $\kappa$, we use KL expansion with $N_2$ terms, which we discuss later. Before describing the priors and the complete hierarchical model for the inverse problem,  the solution for $a_{i,j}$'s for the forward problem is given in the next section.

In Section \ref{sec:forward}, we present a forward solution approach. In Section \ref{sec:inverse}, we write down the priors and hierarchical models for the inverse problem. In Section \ref{sec:post}, we describe the MCMC and variational steps explicitly. Later, we introduce the notion of  the consistency in Section \ref{sec:convergence}. In Section \ref{sec:simulation}, we conduct numerical studies. {We present a simple flowchart in Figure \ref{fig:fc} to help understanding the outline of the paper. }

\begin{figure}
    \centering
    \includegraphics[height=3in]{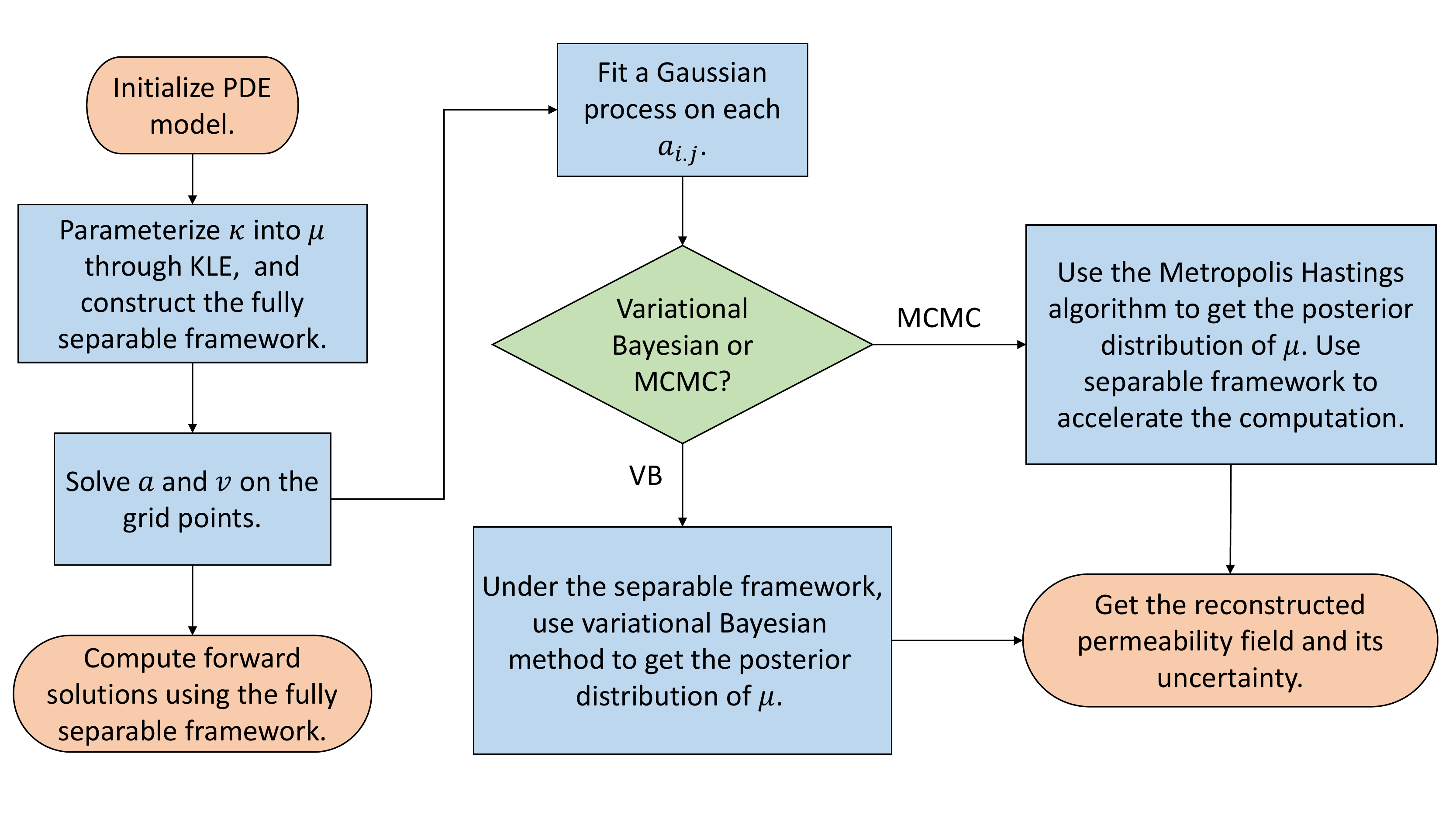}
    \caption{A simple flowchart showing the outline of the paper.}
    \label{fig:fc}
\end{figure}

\section{Forward Solution in Separable Method}
\label{sec:forward}

For PDE model \eqref{model} under the boundary condition a solution can be achieved by solving the variational form,
\[u=\argmin_{u^*} \int_{\Omega_\mu}\int_{\Omega_x}(\frac{1}{2}\kappa|{\nabla} u^*|^2-fv)dxd\mu.\]
Using this form, we can solve for $a_i$ and $v_i$ numerically over a grid in $\Omega_{\mu}$.  When the dimension of  $\Omega_{\mu}$ is bigger then number of required grid over the product space increases exponentially.

To tackle this problem, we further separate the parametric coefficient part as a product of coefficient functions for each individual parameter.  Next, we formulate the fully separable method. Inspired by the idea of Proper Generalized Decomposition, we can introduce following approximation of $u$ in \eqref{sep1}. Let $a=\{a_{i,j}(\cdot)\},\ v=\{v_i(\cdot)\}$, now we try to solve $a$ and $v$ under the variational form.

\subsection{Solving $a$ and $v$}

For derivation of $a_{i,j}$ and $v_{i}$, let us consider the first term, which contains $a_{1,j}$ and $v_{1}$. Using pure greedy optimization, we want to find $a_{1,j}$ and $v_{1}$ that satisfy
\begin{equation*}
 (a_{1,j},v_1) = \argmin_{a_{i,j},v_1}\int_{\Omega_x}\int_{\Omega_\mu}\left(\frac{1}{2}\kappa \left|\nabla \prod^{N_2}_{j=1} a_{1,j}v_1 \right|^2-f \prod^{N_2}_{j=1} a_{1,j}v_1 \right)d\mu dx.
\end{equation*}
Set $a_1 = \prod^{N_2}_{j=1} a_{1,j}$. Then, $v_1$ solves the following equation
\begin{equation}\label{vi}
  -\nabla\cdot\left(\int_{\Omega_\mu}\kappa \prod^{N_2}_{j=1} a_{1,j}^2 d\mu \  \nabla v_1 \right) dx = \int_{\Omega_\mu}f \prod^{N_2}_{j=1} a_{1,j} \ d\mu.
\end{equation}
Similarly, for $a_{1,j}$, we have:
\begin{equation}\label{aij}
  a_{1,j} = \frac{ \int_{\Omega_x} \int_{\Omega_{\mu}\setminus \Omega_j}  f\prod_{k\neq j} a_{1,k}v_1 d\mu' dx }
  {\int_{\Omega_x} \int_{\Omega_{\mu}\setminus \Omega_j} \kappa\prod_{k\neq j} a_{1,k}^2 d\mu' \ |\nabla v_1|^2 dx},
\end{equation}
where $\theta_j\in\Omega_j,\ \Omega_\mu = \Omega_1\bigoplus\Omega_2\bigoplus \cdots\bigoplus\Omega_{N_2}$ and $\mu'\in\Omega_{\mu}\setminus \Omega_j.$
In the previous derivation, we notice that we need to calculate the integration on $\Omega_\mu$. Although $\Omega_\mu$ itself is a high dimensional space, we could compute the integration simply by separating it into $N_2$ subspaces:
\begin{align*}
  \int_{\Omega_\mu}\kappa \prod^{N_2}_{j=1} a_{1,j}^2 d\mu &=p \int_{\Omega_1}\cdots\int_{\Omega_{N_2}} e^{\sum^{N_2}_{j}\theta_{j}\phi_{j}}\prod^{N_2}_{j=1} a_{1,j}^2 d\theta_1\cdots d\theta_{N_2} \\
   &= \prod^{N_2}_{j=1} \int_{\Omega_j}e^{\theta_j\phi_j}a^2_{1,j}d\theta_j
\end{align*}
and
\begin{equation*}
  \int_{\Omega_{\mu}\setminus \Omega_j} \kappa\prod_{k\neq j} a_{1,k}^2 d\mu' = \prod_{k\neq j} \int_{\Omega_k}e^{\theta_k\phi_k}a^2_{1,k}d\theta_k.
\end{equation*}
Here, we use $\kappa=e^{\sum_{j=1}^{N_2}\theta_j\phi_j}$, the form we have for KLE parametrization. We discuss the parametrization in details in next section.

Once  $a_{1,j}$ and $v_1$ are computed, we update $f$ by
\begin{equation*}
f_1 = f+\nabla\cdot(\kappa\nabla\prod^{N_2}_{j=1} a_{1,j}v_1).
\end{equation*}
We state the whole iteration in Algorithm 1.
\begin{algorithm}
\caption{Find $a_{i,j}$ and $v_i$.}
\begin{algorithmic}[1]
\For {$i=1,2,\ldots,N_1$}
\State $a_{i,j}=1, v_i = 1.$
\State $\delta a_i = 1, \delta v_i = 1.$
\While {$\delta a_i>tol_a$ and $\delta v_i>tol_v$}
\State Let $a'_{i,j}=a_{i,j}, v'_i = v_i.$
\State Solve equation \eqref{vi} for $v_i$.
\For {$j=1,2,\ldots,N_2$}
\State Using equation \eqref{aij} to update $a_{i,j}$
\EndFor
\State Update $\delta v_i = ||v_i-v'_i||^2_2 $ and $\delta a_i = \sum_{j}||a_{i,j}-a'_{i,j}||^2$
\EndWhile
\State Update $f = f+\nabla\cdot(\kappa\nabla\prod^{N_2}_{j=1} a_{i,j}v_i)$
\If {$\|f\|<tol_f$}
\State Stop for loop.
\EndIf
\EndFor
\end{algorithmic}
\end{algorithm}

\subsection{GMsFEM framework for forward problem}

In our calculations, we need to solve
the PDE for different permeability fields multiple times.
For this reason, we use the GMsFEM to reduce computational cost.
We use the offline-online procedure to construct GMsFEM coarse spaces.
In the offline stage, we first generate the snapshot spaces based on
different parameters, then get a low dimensional offline space via
model reduction. In the online stage, we use the previous offline basis
function to solve a coarse-grid problem for given $\kappa$.
For details, we refer to \cite{efendiev2013generalized, gao2016application,chung2016adaptive,efendiev2014multilevel,efendiev2015multiscale}.
Note that the GMsFEM becomes particularly important when there are multiple
scales and the problem can be reduced over each coarse block.

\section{Priors and the inverse problem formulation}
\label{sec:inverse}

Given observations on $y$, to estimate and to quantify the uncertainty of $\kappa$ parameterized by {$\mu=\{\theta_1,\dots,\theta_{N_2}\}$}, we specify the priors and the steps of MCMC. We can solve $a_{i,j}$  's over a grid of points and to capture the approximate nature of the representation \eqref{sep1}, we fit a Gaussian process over different values of $\mu$. Under the separation
of $a_i$'s over $\theta_i$'s which results
to a product of Gaussian processes.

\subsection{Fitting a Gaussian Process on function $a_{i,j}$}

Suppose we want to estimate the value of $a_{i,j}$ based on $n_{grid}$  current grid points.
For the function $a_{i,j}(\theta_j)$, from the numerical solution we have the function value $z^{(j)} = (z^{(j)}_1,z^{(j)}_2,\cdots,z^{(j)}_{n_{grid}})$ on $\alpha = (\alpha_1,\alpha_2,\cdots,\alpha_{n_{grid}})$ which are the grids of $\theta_j$. We assume $z_i$ are observed with Gaussian white noise ${\mathcal{N}}(0,\sigma_1^2)$. We set the
covariance matrix $K$ using squared exponential covariance:
\begin{equation*}
  K(\alpha^*,\alpha^{**})=\sigma_a^2\exp\left[-\frac{1}{2}\left (\frac{\alpha^*-\alpha^{**}}{\lambda}\right)^2\right]
\end{equation*}
and
\begin{equation}
a_{i,j}(\cdot)\sim \mathscr{GP}(0,K),
\label{gp_prior}
\end{equation}
where $\lambda$ is the length scale and $\sigma_a$ is the variance.
Then, for a new single input $\alpha_*$, define $\mu_*$ as the mean estimation of $a_{i,j}(\alpha_*)$, and $\sigma_*^2$ as the variance estimation of $a_{i,j}(\alpha_*)$, we have $a_{i,j}(\alpha_*)\sim\mathcal{N}(\mu_*,\sigma_*)$ and
\begin{align}
  \mu_* &= K_{*n_{grid}}(K_{n_{grid}}+\sigma_1^2I)^{-1}z^{(j)} \nonumber \\
  \sigma_*^2 &= K_* - K_{*n_{grid}}p(K_{n_{grid}}+\sigma_1^2I)^{-1}K_{*n_{grid}}^\top+\sigma_1^2 \label{gp},
\end{align}
where $K_{*n_{grid}}$ is the covariance vector between $\alpha_*$ and $(\alpha_1,\alpha_2,\cdots,\alpha_{n_{grid}})$, $K_{n_{grid}}$ is the covariance matrix of vector $(\alpha_1,\alpha_2,\cdots,\alpha_{n_{grid}})$, and $K_* = \sigma_a^2$.
The computational cost for each input point is $\mathcal{O}(N)$ for mean and
$\mathcal{O}(N^3)$ for the variance, with $N=n_{grid}$.

\subsection{Permeability parametrization}

To obtain a permeability field in terms of an optimal $L^2$ basis, we use the KLE.
Considering the random field $\mathscr{Y}( {x},\tilde{\omega})=\log[k( {x},\tilde{\omega})]$, where
$\tilde{\omega}$ represents randomness. We assume a zero mean $E[\mathscr{Y}( {x},\tilde{\omega})]=0$, with a known covariance
operator $R({x},{y})=E\left[\mathscr{Y}( {x})\mathscr{Y}( {y})\right]$. Then,  $\mathscr{Y}(x,\tilde{\omega})$ has following Reproducing Kernel Hilbert Space
(RKHS) representation
\[
\mathscr{Y}( {x},\tilde{\omega})=\sum_{k=1}^{\infty} \mathscr{Y}_k(\tilde{\omega}) \Phi_k( {x}),\qquad
\]
with
\begin{equation*}
\mathscr{Y}_k(\tilde{\omega})=\int_\Omega\mathscr{ Y}(x,\tilde{\omega})\Phi_k(x)dx.
\end{equation*}
The functions $\{\Phi_k(x)\}$ are eigenvectors of the covariance operator $R(x,y)$, and form
a complete orthonormal basis in $L^2(\Omega)$,that is
\begin{equation}\label{eig}
\int_{\Omega} R( {x}, {y})\Phi_k( {y})d {y}=\lambda_k\Phi_k( {x}),
\qquad k=1,2,\ldots,
\end{equation}
where $\lambda_k=E[\mathscr{Y}_k^2]>0$, $E[\mathscr{Y}_i\mathscr{Y}_j]=0$ for
all $i\neq j$. Let $\theta_k=\mathscr{Y}_k/\sqrt{\lambda_k}$. Hence, $E[\theta_k]=0$ and $E[\theta_i\theta_j]=\delta_{ij}$ and  we have
\begin{equation}\label{KLE}
\mathscr{Y}( {x},\tilde{\omega})=\sum_{k=1}^{\infty} \sqrt{\lambda_k}\theta_k(\tilde{\omega})\Phi_k( {x}),
\end{equation}
where $\Phi_k$ and $\lambda_k$ satisfy \eqref{eig}.  We truncate the KLE \eqref{KLE} to a finite number of terms and keep the leading-order terms (quantified by the magnitude
of $\lambda_k$), and  capture most of the energy of the stochastic process $\mathscr{Y}(x,
\tilde{\omega})$. For an $N$-term KLE approximation
\begin{equation*}
\mathscr{Y}^N=\sum_{k=1}^{N}\sqrt{\lambda_k}\theta_k\Phi_k,
\end{equation*}
the energy ratio of the approximation is defined by
\begin{equation*}
e(N):=\frac{E\|\mathscr{Y}^N\|^2}{E\|\mathscr{Y}\|^2}=\frac{\sum_{k=1}^N\lambda_k}{\sum_{k=1}^{\infty}\lambda_k}.
\end{equation*}
If the eigenvalues $\{\lambda_k\}$ decay very fast, then the truncated KLE with the
first few terms would be a good approximation of the stochastic process $\mathscr{Y}(x,\omega)$ in the $L^2$ sense.
In our simulations, we use the prior $\theta_k \sim {\mathcal{N}}(0,1)$.

\subsection{Hierarchical Model}
From the likelihood equation \eqref{ermodel} and the priors in  \eqref{KLE}
and \eqref{gp_prior}, the full hierarchical model  can be written. Writing $\tilde{u} = \sum^{N_1}_{i=1} \left(\prod^{N_2}_{j=1} a_{i,j}(\theta_j)\right)v_i(x)$, as in the fully separable representation given in \eqref{sep1}, we have
\begin{align}
y &= F(\kappa)+\epsilon_f, \nonumber \\
\epsilon_f&\sim {\mathcal{N}}(0,\sigma^2_f), \nonumber \\
a_{i,j}(\cdot)&\sim \mathscr{GP}(0,K), \nonumber \\
log(\kappa( {x},\tilde{\omega}))&=\sum_{k=1}^{\infty} \sqrt{\lambda_k}\theta_k(\tilde{\omega})\Phi_k( {x}),\nonumber \\
\theta_k &\sim {\mathcal{N}}(0,1).
\label{hr_model}
\end{align}

\section{Posterior Calculation}
\label{sec:post}

Suppose we have $M$ observations $y_1,y_2,\cdots,y_M \in R$ of $u$ at the locations $x_1,x_2,\cdots,x_M\in \Omega_x$. Then, we define vector $V_i = \left(v_i(x_1),v_i(x_2),\cdots,v_i(x_M)\right)^\intercal$, where $v_i$ is the solution of \eqref{vi}, and let $y=(y_1,y_2,\cdots,y_M)^\intercal$. We also assume that each observation contains Gaussian white noise ${\mathcal{N}}(0,\sigma_{y}^2)$.
Based on the separated method and Bayes framework the posterior is given by:
\begin{equation}
  \Pi(\mu,a|y)\varpropto p(y|a)\cdot \pi(a|\mu) \cdot \pi(\mu).
  \label{post_den}
\end{equation}

\subsection{Variational Bayesian Method}

We present a brief overview of the variational Bayesian method. For the observation  $Y$,  parameter ${\bf \Theta}$ and prior  $\Pi({\bf \Theta})$ on it,  let the  joint distribution and posterior be   ${p(Y,\Theta)}$ and  ${\Pi({\bf\Theta}|Y)}$, respectively. Then
\begin{equation*}
\log p(Y) = \int \log\frac{p(Y,\Theta)}{q({\bf \Theta})} q({\bf \Theta}) d({\bf \Theta})+KL(q({\bf \Theta}),\Pi({\bf\Theta}|Y)),
\end{equation*}
for any density $ q({\bf \Theta})$. Here $KL(p,q)=E_p(\log\frac{p}{q})$, the Kullback-Leibler distance between $p$ and $q$.
Thus,
\begin{equation}
\label{vrkl}
 KL(Q({\bf \Theta}),p(Y,{\bf\Theta})) =   KL(Q({\bf \Theta}),\Pi({\bf\Theta}|Y))-\log p(Y) .
\end{equation}
Given $Y$, we minimize  $ KL(Q({\bf \Theta}),p(Y,{\bf\Theta}))$ under separability.
Minimization of the L.H.S of \eqref{vrkl} analytically may not be possible in general and therefore, to simplify the problem, it is assumed that the parts of ${\bf \Theta}$ are conditionally independent given $Y$. That is
\[ Q({\bf \Theta})=\prod_{i=1}^s q(\eta_i)\]
and $\cup_{i=1}^s \eta_i={\bf \Theta}$ is a partition of the set of parameters ${\bf \Theta}$. Minimizing under the separability assumption, an approximation of the posterior distribution is computed.   Under this assumption of minimizing L.H.S of \eqref{vrkl} with respect to $q(\eta_i)$, and  keeping the other $q( \eta_j ), j\neq i$ fixed, we develop the  following mean field approximation equation:

\begin{equation}
q(\eta_i) \propto \exp( E_{-i}\log(p(Y,{ \bf \Theta})),
\label{vrupdate}
\end{equation}
where the subscript $E_{-i}$ stands for expectation with respect to variables other than those included in  $\eta_i$.
The variational solution is exact and converges to the KL minimizer    in the proposed class rapidly.

Finally for our case, the posterior distribution of parameter $\mu$ and $a$ given observations $y$ is approximated by a variational distribution:
\begin{equation*}
  P(\mu,a|y)=Q(\mu,a).
\end{equation*}
For a fixed index $j$, $\theta_j$ and $a_{i,j}$ can
 be highly correlated due to Gaussian process fitting.
We have the following from equation \eqref{gp}
\begin{align*}
  E(a_{i,j}) &= K_{\theta_jn}(K_n+\sigma_1^2I)^{-1}z^{(j)} \\
  Var(a_{i,j}) &= K_{\theta_j} - K_{_{\theta_j}n}(K_n+\sigma_1^2I)^{-1}K_{\theta_jn}^\intercal+\sigma_1^2.
\end{align*}
We group these variables based on index $j$:
\begin{equation*}
  \mathcal{G}_j = \{\theta_j,a_{1,j},a_{2,j},\cdots,a_{N_1,j}\}  \text{ for } j=1,2,\cdots,N_2
\end{equation*}
and then we factorize the variational distribution into these groups,
\begin{equation}
  Q(\mu,a)= \prod_{j=1}^{N_2} q(\theta_j,a_{1,j},\cdots,a_{N_1,j}).
\end{equation}
Minimizing the Kullback-Leibler divergence of $P$ from $Q$, the best approximation of each $q_j$ is
\begin{equation}\label{VBbasic}
  \ln q(\mathcal{G}_j)=E_{\theta_k,a_{i,k},k\neq j}\left[\ln P(\mu,a,y)\right]+\text{constant}.
\end{equation}
In the following part we write $E_{\theta_k,a_{i,k},k\neq j}$ as $E_{-j}$ in abbreviation.

Next, our goal is to give
a complete form of Equation \eqref{VBbasic}.
For this, we suppose we have $M$ observations $y_1,y_2,\cdots,y_M \in R$ on $x_1,x_2,\cdots,x_M\in \Omega_x$, then we shrink $v_i(x)$ into an $M$ dimension vector
\begin{equation*}
  V_i = \left(v_i(x_1),v_i(x_2),\cdots,v_i(x_M)\right)^\intercal
\end{equation*}
and let $y=(y_1,y_2,\cdots,y_M)^\intercal$. We also assume that each observation contains Gaussian white noise ${\mathcal{N}}(0,\sigma_{y}^2)$.

Furthermore, we use Gaussian distribution as prior distribution for $\theta_j$,
\begin{equation*}
  \theta_j \sim {\mathcal{N}}(\theta_0,\sigma_0^2).
\end{equation*}
Then, we have
\begin{align*}
  \ln q(\mathcal{G}_k) = {} & E_{-k}\left[\ln P(\mu,a,y)\right]+\text{constant} \\
   =  {} & E_{-k}\left[\ln P(y|a)+\ln P(a|\mu)+\ln P(\mu)\right]+\text{constant} \\
   =  {} & E_{-k}\left[-\frac{\norm{y-\sum^{N_1}_{i=1} \left(\prod^{N_2}_{j=1} a_{i,j}\right)V_i}_2^2}{2\sigma_y^2}\right]  \\
    & + E_{-k}\left[\sum_{i=1}^{N_1}\sum_{j=1}^{N_2}\ln\frac{\exp\left(-\frac{1}{2\sigma_{i,j}^2(\theta_j)}|a_{i,j}
   -\mu_{i,j}(\theta_j)|^2\right)}{\sqrt{2\pi}\sigma_{i,j}(\theta_j)}\right] \\
   & + E_{-k}\left[\sum_{j=1}^{N_2}\ln\frac{\exp\left(-\frac{1}{2\sigma_{0}^2}(\theta_j-\theta_{0})^2\right)}
   {\sqrt{2\pi}\sigma_0}\right] + C.
\end{align*}

Define $R_i = \prod_{j\neq k}{a_{i,j}}$ and re-write the above equality as
\begin{align*}
  \ln q(\mathcal{G}_k) = {}& E_{-k}\left[-\frac{1}{2\sigma_y^2}\norm{y-\sum^{N_1}_{i=1} a_{i,k}R_i V_i}_2^2\right] \\
   &- \sum_{i=1}^{N_1}\frac{1}{2\sigma_{i,k}^2(\theta_k)}|a_{i,k}
   -\mu_{i,k}(\theta_k)|^2 - \ln\prod_{i=1}^{N_1} \sigma_{i,k}(\theta_k) \\
   &- \frac{1}{2\sigma_0^2}(\theta_k-\theta_0)^2 + C.
\end{align*}

We need to calculate the marginal distribution of $\theta_j$.
Let $A_k = (a_{1,k},\cdots,a_{N_1,k})^\intercal\in \mathbb{R}^{N_1}$. Then
\begin{align*}
  q(\theta_k) &= \int_{R^{N_1}} \ q(\theta_k,A_k) \ dA_k\\
   &= \int_{R^{N_1}} e^{-A_k^\intercal \Sigma(\theta_k) A_k + \beta^\intercal(\theta_j) A_k } \ dA_k \cdot\frac{1}{\prod_i \sigma_{i,k}(\theta_k)}
        \cdot e^{-\sum_{i=1}^{N_1}\frac{\mu_{i,k}^2(\theta_k)}{2\sigma_{i,k}^2(\theta_k)}-\frac{(\theta_k-\theta_0)^2}{2\sigma^2_0}} \cdot C,
\end{align*}
where
\begin{align*}
  \Sigma_{m,m}(\theta_k) &= \frac{1}{2\sigma^2_y}E(R_m^2)V_m^\intercal V_m +  \frac{1}{2\sigma_{m,k}^2(\theta_k)} \quad \text{for} \quad m=1,2,\cdots,N_1,\\
  \Sigma_{m,n}(\theta_k) &= \frac{1}{2\sigma^2_y}E(R_mR_n)V_m^\intercal V_n
  \quad \text{for} \quad m\neq n, \\
  \beta_m(\theta_k) &= \frac{1}{\sigma_y^2}E(R_m)V_m^\intercal y
            + \frac{\mu_{m,k}(\theta_k)}{\sigma_{m,k}^2(\theta_k)}
            \quad \text{for} \quad m=1,2,\cdots,N_1.
\end{align*}

Note that
\begin{equation*}
  \int_{R^{N_1}} e^{-A_k^\intercal \Sigma A_k + \beta^\intercal A_k } \ dA_k = \sqrt{\frac{\pi^{N_1}}{\det\Sigma}}\exp(\frac{1}{4}\beta^\intercal\Sigma^{-1}\beta).
\end{equation*}
Therefore,
\begin{equation}\label{theta}
  q(\theta_k) \propto q^*(\theta_k) = \frac{1}{\sqrt{\det\Sigma} \cdot \prod_i \sigma_{i,k}(\theta_k)}
  \exp\left(\frac{1}{4}\beta^\intercal\Sigma^{-1}\beta-
  \frac{1}{2\sigma_0^2}(\theta_k-\theta_0)^2-\sum_{i=1}^{N_1}\frac{\mu_{i,k}^2(\theta_k)}{2\sigma_{i,k}^2(\theta_k)}\right).
\end{equation}
We can compute the joint for $\theta_k,a_{p,k}$,
\begin{multline}
   q^*(\theta_k,a_{p,k}) =
  \frac{1}{\sqrt{\det\Sigma} \cdot \prod_i \sigma_{i,k}(\theta_k)}
  \\
  e^{\left(\frac{1}{4}\beta^\intercal\Sigma^{-1}\beta-c_p(\theta_k)
  -\frac{(\theta_k-\theta_0)^2}{2\sigma^2_0}-
  \frac{1}{2\sigma_y^2}\left(a_{p,k}^2E(R_p^2)v_p^\intercal v_p\right)
  \right)},
\end{multline}
where
\begin{align*}
  \Sigma_{m,m} &= \frac{1}{2\sigma^2_y}E(R_m^2)V_m^\intercal V_m +  \frac{1}{2\sigma_{m,k}^2(\theta_k)} \quad \text{for} \quad m\neq p, \\
  \Sigma_{m,n} &= \frac{1}{2\sigma^2_y}E(R_mR_n)V_m^\intercal V_n
  \quad \text{for} \quad m\neq n,m\neq p,n\neq p, \\
  \beta_m &=p \frac{1}{\sigma_y^2}E(R_m)V_m^\intercal y
            + \frac{\mu_{m,k}(\theta_k)}{\sigma_{m,k}^2(\theta_k)}
            + \frac{1}{\sigma_y^2}a_{p,k}E(R_pR_m)V_p^\intercal V_m,\\
   c_p(\theta_k)&=     \sum_{i=1}^{N_1}\frac{\mu_{i,k}^2(\theta_k)}{2\sigma_{i,k}^2(\theta_k)}
  +\frac{(a_{p,k}-\mu_{p,k})^2}{2\sigma_{p,k}^2(\theta_k)} 
  - \frac{1}{\sigma_y^2}a_{p,k}E(R_p)v_p^\intercal y.
\end{align*}

The details of the expansions are given in Appendix A1.
Next, we write down the algorithm for Variational Bayesian method
in Algorithm 2.

\begin{algorithm}
\caption{Variational Bayesian Method.}
\begin{algorithmic}[1]
\State Set initial value $\theta_j = 0$ for all $j$.
\State Calculate $E(a_{i,j})$ from \eqref{gp} for all $i,j$.
\State Calculate $E(R_m),\ E(R_m^2)$ and  $E(R_mR_n)$ by \eqref{RMN}.
\State Set $\delta \theta = 1, \delta a = 1.$
\While {$\delta a>tol_a$ and $\delta \theta>tol_\theta$}
\State Let $E(a')= E(a), E(\theta') = E(\theta).$
\For {$j=1,2,\cdots,N_2$}
\State Use equations \eqref{vtheta} to update $E(\theta_j)$ and $\sigma_{\theta_j}$.
\For {$i=1,2,\cdots,N_1$}
\State Use equations \eqref{vaij} to update $E(a_{i,j})$ and $\sigma_{a_{i,j}}$.
\EndFor
\State Update $E(R_m),\ E(R_m^2)$ and  $E(R_mR_n)$ by \eqref{RMN}.
\EndFor
\State Update $\delta \theta = \sum_j||E(\theta_j) - E(\theta_j')||^2_2 $ and $\delta a = \sum_{i,j}||E(a_{i,j})-E(a'_{i,j})||^2$
\EndWhile
\State Use equations \eqref{evkappa} to get $E(\kappa)$ and $\sigma(\kappa)$.
\end{algorithmic}
\end{algorithm}

\subsection{MCMC method}

From the posterior density given in \eqref{post_den}, we use the Metropolis-Hastings Gibbs sampling  algorithm from the  following set up:
\begin{align*}
  p(y|a) &\varpropto \exp\left(-\frac{1}{\sigma_{y}^2}\norm{y-\sum^{N_1}_{i=1} \left(\prod^{N_2}_{j=1} a_{i,j}\right)V_i}^2\right) \\
  \pi(a|\mu) &\varpropto \prod^{N_1}_{i=1} \prod^{N_2}_{j=1} \ \exp\left(-\frac{1}{\sigma_{\theta_j}^2}|a_{i,j}-\mu_{\theta_j}|^2\right) \\
  \pi(\mu) &\varpropto \prod^{N_2}_{j=1} \pi(\theta_j),
\end{align*}
where $y$ are observations, $a=\{a_{i,j}\},\ \mu=(\theta_1,\cdots,\theta_{N_2}), \ \mu_{\theta_j}$ and $\sigma_{\theta_j}^2$ are mean and variance estimation from the Appendix.

\section{Some convergence analysis}
\label{sec:convergence}

To show the convergence of the proposed method, we need to show the convergence of the forward and the inverse problems. For the forward problem, we show that the proposed separable  representation can adequately approximate $u$. The greedy algorithm for the variational form converges and the solution converges to $u$ in some appropriate metric.
For the inverse problem solution, given the observed values, under some conditions, the posterior distribution of $\kappa$ (or $\mu$) converges to the true parameter value in distributional sense.

\subsection{Convergence of the forward problem}
In Section \ref{sec:forward}, we introduced fully separated framework for parameter-dependent elliptic partial differential equations and provided an iterative algorithm.
Next, we state some convergence results following \cite{le2009results}
for this algorithm.
We restrict ourselves to the simplest case. Find $a_i$ in $L^2(\Omega_\mu)$ 
and $v_i$ in $H^1_0(\Omega_x)$ 
, such that $u = \sum_{i\geq 1}a_i(\mu)v_i(x)$ is the solution of
\begin{align}\label{RA1}
  -\nabla\cdot(\kappa(x,\mu)\nabla u(x,\mu)) &= f \textrm{ in } \Omega_x, \nonumber \\
  u|_{\partial \Omega_x} &= 0,
\end{align}
given $0<c_1\leq\kappa(x,\mu)\leq c_2<\infty$ in the bounded space $\Omega$, where $\mu\in\Omega_{\mu},\ x\in\Omega_x$, and $\Omega = \Omega_\mu \times \Omega_x,\ \Omega_\mu\in R^{N_2},\ \Omega_x\in R^2$.
Following \cite{le2009results}, we introduce the tensor product Hilbert space
\begin{equation}\label{eq:Gamma}
  \Gamma(\Omega)=L^2(\Omega_\mu)\bigotimes H^1_0(\Omega_x)
\end{equation}
with inner product:
\begin{equation*}
  \langle u,v \rangle = \int_{\Omega} \kappa(x,\mu) \nabla u(x,\mu) \cdot \nabla v(x,\mu) \quad u,v\in\Gamma
\end{equation*}
and the associated norm
\begin{equation*}
  \|u\|^2 = \int_{\Omega} \kappa(x,\mu) |\nabla u(x,\mu)|^2.
\end{equation*}
Note that
\begin{equation}\label{eq:av}
    (a_n,v_n)= \argmin_{(a,v)\in L^2(\Omega_\mu)\times H^1_0(\Omega_x)} \int_{\Omega}\frac{1}{2}\kappa|\nabla (a\otimes v)|^2-\int_{\Omega}f_{n-1}\cdot(a\otimes v).
\end{equation}
Now we introduce $u_n$ satisfying
\begin{align}\label{eq:un}
    -\nabla\cdot(\kappa\nabla u_n) &= f_n \textrm{ in } \Omega_x, \nonumber \\
    u_n|_{\partial \Omega_x} &= 0,
\end{align}
then we have
\begin{equation}
  u_n = u_{n-1}-a_n\otimes v_n
\end{equation}
therefore $u_n = u_0-\sum^n_{i=1}a_i\otimes v_i$.
We denote the tensor product $a\otimes v$ as
$$ a\otimes v(\mu,x) := a(\mu)\cdot v(x). $$
Following \cite{le2009results}, we can prove the following lemma.
\begin{theorem}
Assume that $(a_n,v_n)$ satisfies \eqref{eq:av}. Denote the energy at iteration $n$ as
\begin{equation*}
  E_n = \frac{1}{2}\int_\Omega\kappa|\nabla (a_n\otimes v_n)|^2-\int_{\Omega}f_{n-1}\ a_n\otimes v_n,
\end{equation*}
then we have
\begin{equation}\label{eq:Con1}
  \lim_{n\rightarrow\infty}E_n=\lim_{n\rightarrow\infty}\|a_n\otimes v_n\|=0.
\end{equation}
Moreover,
\begin{equation}\label{eq:Con2}
  \lim_{n\rightarrow\infty}u_n = 0 \text{ in } \Gamma,
\end{equation}
{where $\Gamma$ is defined in \eqref{eq:Gamma},} and
\begin{equation}\label{eq:Con3}
  \lim_{n\rightarrow\infty}f_n = 0 \text{ in } \Gamma^*.
\end{equation}
where $\Gamma^*$ is the dual space of $\Gamma$.
\label{conv1}
\end{theorem}
\begin{proof}
Given in Appendix A2.
\end{proof}
Based on this convergence, one can estimate the difference
between the posterior that uses truncated series and the posterior
that uses all the terms in appropriate norms and show that this
difference is small independent of the
dimension of the problem
\cite{mondal2010bayesian, wei2012reduced,stuart2010inverse}.

\subsection{Consistency of Bayesian Method}
In this section, we discuss and derive the consistency of Bayesian methodology.
The main idea is that if we have more and more observations, the posterior distribution will concentrate around true data generating  distribution (see e.g.,
\cite{ghosh2003bayesian}). Even though, we have observations on finitely many grids, posterior consistency is a desirable property for an estimator to have.
If the proposed model is true and if we have enough prior mass around the true parameter then the likelihood will pull the posterior distribution around the true parameter value.  Here, we assume the data generating model, is given in \eqref{ermodel}, posterior in \eqref{posterior} and the priors are \eqref{gp_prior} and \eqref{KLE}. We also restrict $\theta_j$'s to a compact set.  In this setting, we suppose, $f_{\kappa}$ is the density for the $\kappa$ in \eqref{ermodel} and $\kappa^*$ is the true permeability field.
Furthermore, let $f^*_x(y)$ be the true data generating density at $x \in \Omega_x$, where $y\in \Omega_y$, a subset of Euclidean space. We assume we have observation in $x_1,\dots,x_M$ under the measure $H(x)$ on $\Omega_x$. For a fixed design point, we denote the empirical measure $H_M$.  Then, we have the following consistency theorem (cf. \cite{ghosh2003bayesian}).
\begin{theorem}
For  the model and prior setting given in equation \eqref{hr_model}, for any  neighborhood $U_\epsilon$ of  $f^*_x(y)$ given as
$\int_{\Omega_x}\int_{\Omega_y}(\sqrt{f_x(y)}-\sqrt{f^*_x(y)})^2dydH(x)<\epsilon$, we have  $\Pi(U_\epsilon|data)\rightarrow 1$  with probability 1 as $M\rightarrow \infty$.
\label{consistency1}
\end{theorem}

The details are given in Appendix A2. The proof uses our separation argument and additional conditions on the priors.
 Under the given prior, we have positive probability around any $L_\infty$ neighborhood of $a_i(\{\theta\})$. This condition ensures positive probability in any Kullback-Leibler neighborhood of $f^*$. Then, given the ratio of the  likelihood under true density and any other density incerases to infinity exponentially with $M$, the result follows.

If the observations points are fixed design points given by an empirical measure $H_M$,  then  we will have posterior concentration on  $U_\epsilon$, that is the set given by,  $\int_{\Omega_x}\int_{\Omega_y}(\sqrt{f_x(y)}-\sqrt{f^*_x(y)})^2dydH_M(x)<\epsilon$. For the convenience, here we assume that the design points are coming from underlying $H$ on $\Omega_x$, with density $h(x)$.

\section{Numerical Results}
\label{sec:simulation}

In this section, we first apply our fully separated method for the forward problem. Then under the structure of fully separated method, we use MCMC method and Variational Bayesian Method on inverse problem and Uncertainty Quantification.

\subsection{Results for Forward Problem}

In order to calculate $a_{i,j}$ and $v_i$, we first need to define a range on our parameter $\theta_j$. Based on our prior information $\theta_j \sim {\mathcal{N}}(0,1)$, we simply let $\theta_j\in [\theta_{\min},\theta_{\max}]$ and $\theta_{\min}=-5, \ \theta_{\max}=5.$ When $N_2\leq 200$, this range will cover more than $99.99\%$ of $\theta$ according to our standard Gaussian prior. We choose $21$ grid points evenly located in interval $[\theta_{\min},\theta_{\max}]$, which are the grid points of our discrete function $a_{i,j}$.

We test our model for $10$ and $30$ KL basis separately. For $N_2 = 10$, we choose $N_1 = 10$. For $N_2 = 30$, we choose $N_1=50$. We use fully separated method to get $a$ and $v$, thus given parameters we have $u_{fs}$. The calculation is performed on $50\times50$ grid.
To test the result, we randomly sample $100$ samples of $\mu=(\theta_1,\cdots,\theta_{N_2})$ on grid points based on the prior, then use FEM method to calculate $u_{fem}$ as the accurate result. Finally, we compute relative $L^2$ error and relative energy error between $u_{fem}$ and $u_{fs}$,
\begin{table}
\centering
\begin{tabular}{|c|cc|}
  \hline
 Error  & $N_2 = 10$ & $N_2 = 30$\\
\hline
  mean relative $L^2$ error & 0.0260\% & 0.0840\% \\
  mean relative energy error& 0.4728\% & 0.4828\% \\
  max relative $L^2$ error & 0.0882\% & 0.1942\% \\
  max relative energy error & 0.7256\% & 0.8772\% \\
\hline
\end{tabular}
\caption{Error in forward model.}
\label{errf}
\end{table}
We can see that all the relative errors are under $1\%$ (see Table \ref{errf}). Therefore, our fully separated method performs well for the forward problem.

\subsection{Results for Inverse Problem}

Given $M$ observations of $u$, we would like to estimate and to quantify the uncertainty of $\kappa$ parameterized by $\mu=(\theta_1,\dots,\theta_{N_2})$. Let $y=(y_1,\cdots,y_M)$ be observations. We assume the locations of these observations are evenly distributed on $\Omega_x$ as Figure \ref{fig:obs} shows. The number of observations could be small (9 points) or large (100 points) in our model. Note that on our $50\times 50$ grid, even 100 observation points only contains 4\% of total output data, which leads to ill-posedness of the inverse problem. We use Markov Chain Monte Carlo and Variational Bayesian Method separately to reconstruct $\kappa$.
We use KL expansion to generate the reference field $\kappa_{ref}$ based on our prior. In this example, we choose isotropic random Gaussian field with correlation function $C(\alpha,\beta)$ given by
$C(\alpha,\beta)=\sigma_{gf}^2\exp(-\frac{|\alpha-\beta|^2}{2l_0^2})$, where $l_0$ is the correlation lengths in both direction and $\sigma_{gf}$ determines the variation of the field. Here we let $\sigma_{gf}=0.2,\ l_0=0.1$.

\begin{figure}
    \centering
    \begin{subfigure}[t]{0.5\textwidth}
        \centering
        \includegraphics[height=2in]{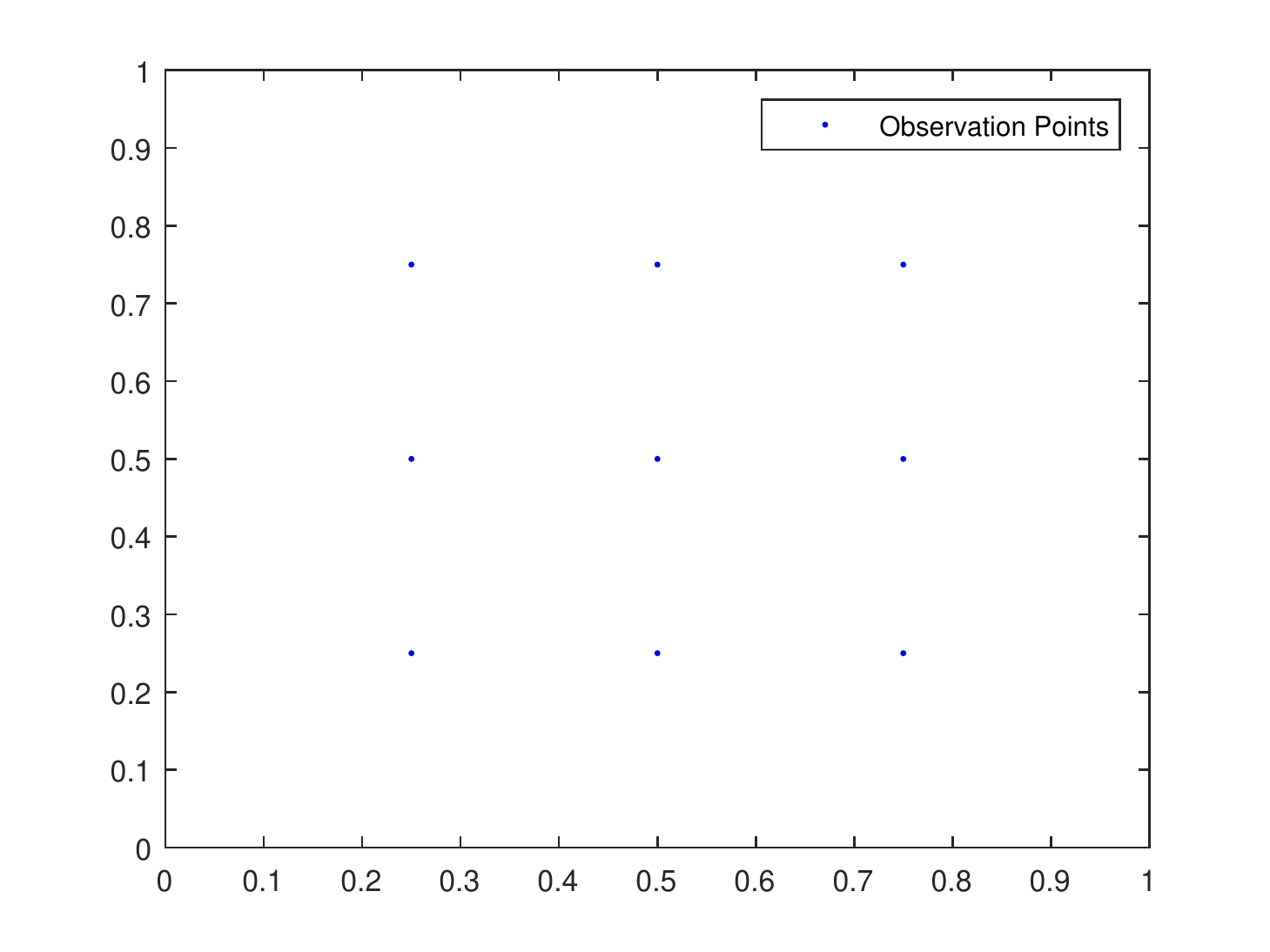}
        \caption{Fewer observations.}
    \end{subfigure}%
    \begin{subfigure}[t]{0.5\textwidth}
        \centering
        \includegraphics[height=2in]{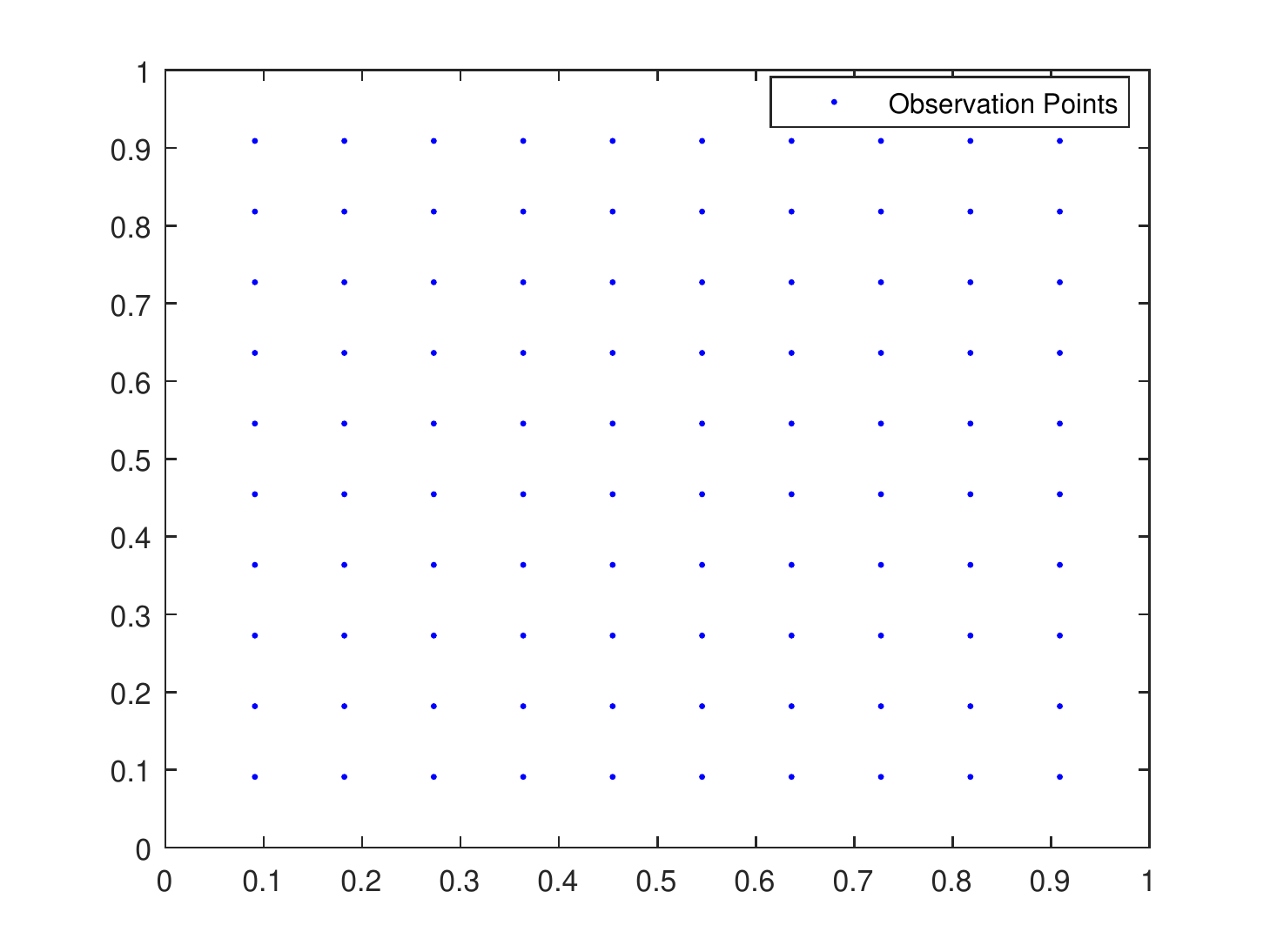}
        \caption{More observations.}
    \end{subfigure}
    \caption{Locations of fewer and more observation points.}
    \label{fig:obs}
\end{figure}

We set our Gaussian process parameters as $\sigma_a = 1$, $\sigma_1 = 10^{-3}$, $\lambda=1$. We let $\sigma_y$ be 1\% of the mean observed value. For MCMC, we use $10^7$ iterations and $10^5$ burn-in. For variational Bayesian Method, we use trapezoidal method to do numerical integration and split the domain into 50 intervals, and assume independence between different $a_{i,j}$ to ease our computational complexity. All our examples are performed by Matlab on a 16-Core CPU.

In the case of fewer (9 points, 0.36\% of total output) observations, due to the lack of observations, the problem become very ill-posed, which makes it difficult to reconstruct complicated field, thus we use a rather smooth field as reference field. In the case of more (100 points, 4\% of total output) observations, we try to reconstruct more complicated field. For the smooth case, we use $(a,v)$ of $N_2=10$, and for we use $(a,v)$ of $N_2=30$,  as in the forward problem example. The result is shown in Figure \ref{fig:ipresult}. {We can see that both the VB and MCMC methods give good reconstructions on the reference log permeability fields. Due to the lack of observation points in the left column, some details of the field cannot be reconstructed well, while the results in the right column with more observations have shown more details of the reference field. Moreover, from (b), (d) and (f) of Figure 3, it's clear that the variational Bayesian method gives a smoother and more accurate result over the reconstructed field than the result from MCMC method, thus performs better in this case.}

{Define $E\mu$ as the posterior mean of $\mu$ in each iteration such that $E\mu = (E\theta_1, E\theta_2, \dots, E\theta_{N_2})$, we compare the trace of all $E\theta_j$ from MCMC and variational Bayesian method in Figure \ref{fig:trace} for both fewer and more observations. Variational Bayes converges over a smooth path over first few iterations. With the increase in the dimensions of $\mu$ won't increase the steps VB method need to converge, while the number of steps that MCMC needs clearly depends largely on the dimensions of $\mu$, which makes variational Bayesian method a preferred choice when dealing with high-dimensional problems.} Finally, we show the distributions of parameters in the complicated case as we draw the posterior and prior of first 8 $\theta_j$ from KLE in Figure \ref{fig:theta_distribution}.

About the computational cost for variational method, we can use parallel processing for $i=1,2,\cdots,N_1$ inside each iterations (line 9 to 11 in Algorithm 2), but for MCMC we could not do parallel processing on a single chain since it is a Markov process. We do a runtime comparison between variational Bayesian method and MCMC method in Table \ref{rt1}. Similarly, we could perform parallel processing in line 8-10 of Algorithm 1 on the fully separated forward problem framework. { Moreover, when we use MCMC method, computational time inside each iteration with fully separated framework is much faster than the one without our framework (eg. 0.0025s compare to 0.022s for complicated case here), which makes a significant difference when the number of iterations is very large in most cases.}

The variational Bayes estimator provides an adequate posterior approximation, while converging very fast, even in complicated cases. {We define $\delta\mu$ as the $L^2$ norm of the differences of $E\mu$ between adjacent variational Bayesian iterations.} A convergence diagnostic can be found in Figure \ref{fig:VBcon}, where the fast convergence can be noticed.

\section{Discussion}

In this paper,
we study the uncertainty quantification in inverse problems. We consider
a multiscale diffusion problem and use separation of stochastic and spatial
variables. A new separation technique is introduced.
Under separable representation, proposed Bayesian approaches provide an accurate and fast solution of the inverse problem and  provide measure of uncertainty of the parameters given the data. A fast variational Bayes based posterior approximation algorithm has been introduced which produces further computational efficiency and an accurate estimation of the underlying field.

\begin{figure}
    \centering
    \begin{subfigure}[t]{0.4\textwidth}
        \centering
        \includegraphics[width=\textwidth]{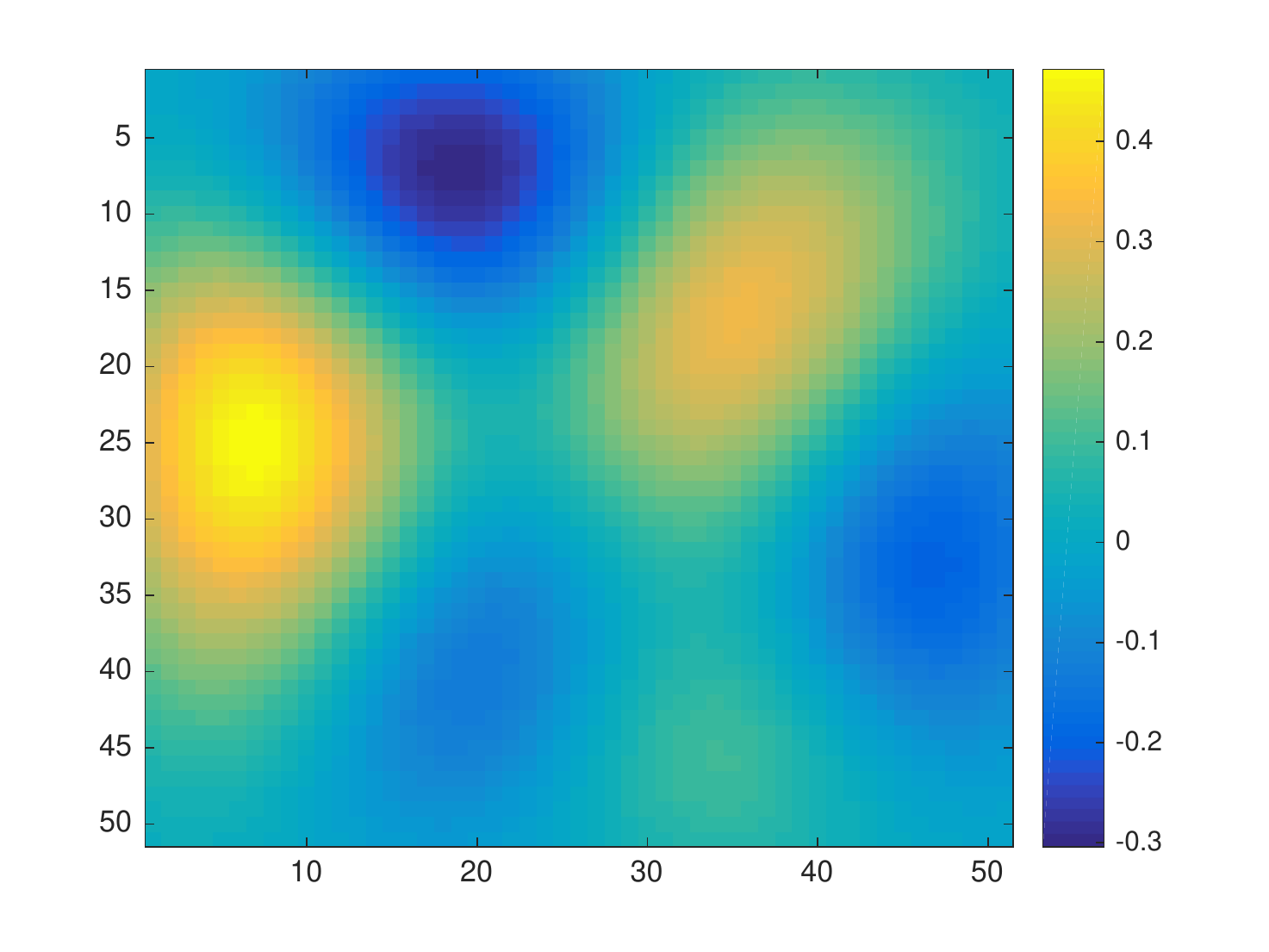}
        \caption{Reference log permeability field (smooth case).}
    \end{subfigure}\quad
    \begin{subfigure}[t]{0.4\textwidth}
        \centering
        \includegraphics[width=\textwidth]{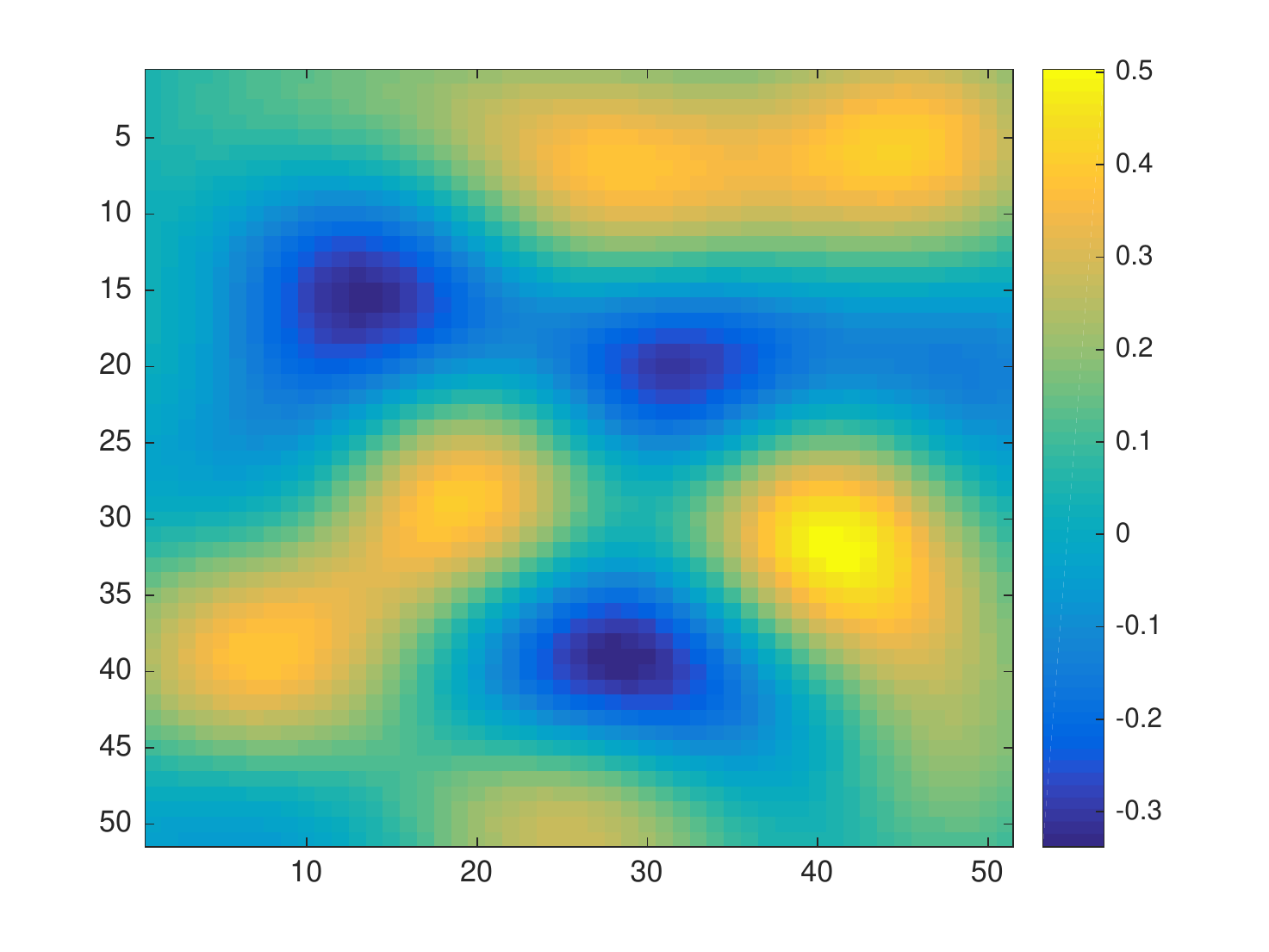}
        \caption{Reference log permeability field (complicated case).}
    \end{subfigure}\quad
    \begin{subfigure}[t]{0.4\textwidth}
        \centering
        \includegraphics[width=\textwidth]{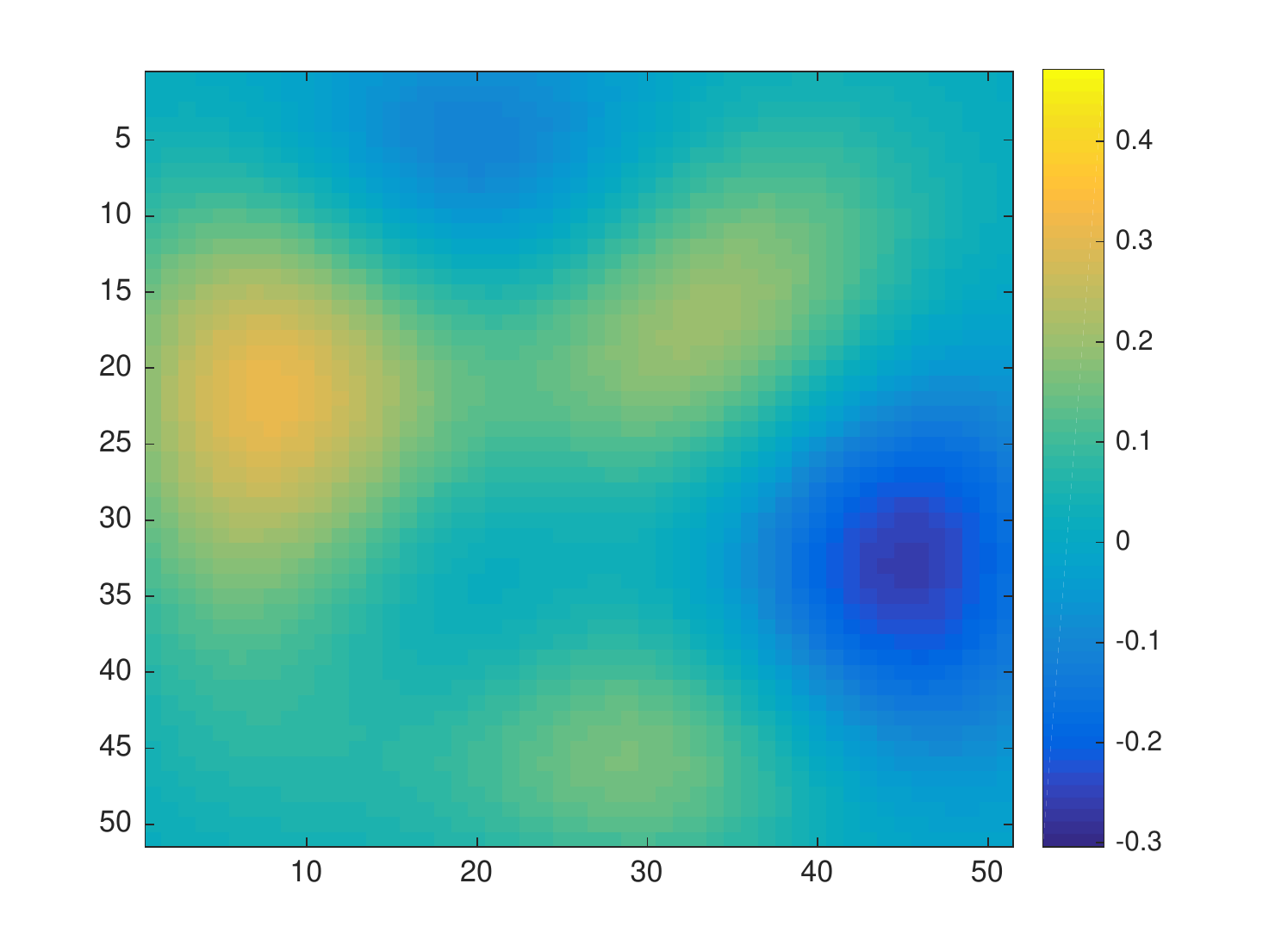}
        \caption{Mean log permeability field by VB method (smooth case).}
    \end{subfigure}\quad
    \begin{subfigure}[t]{0.4\textwidth}
        \centering
        \includegraphics[width=\textwidth]{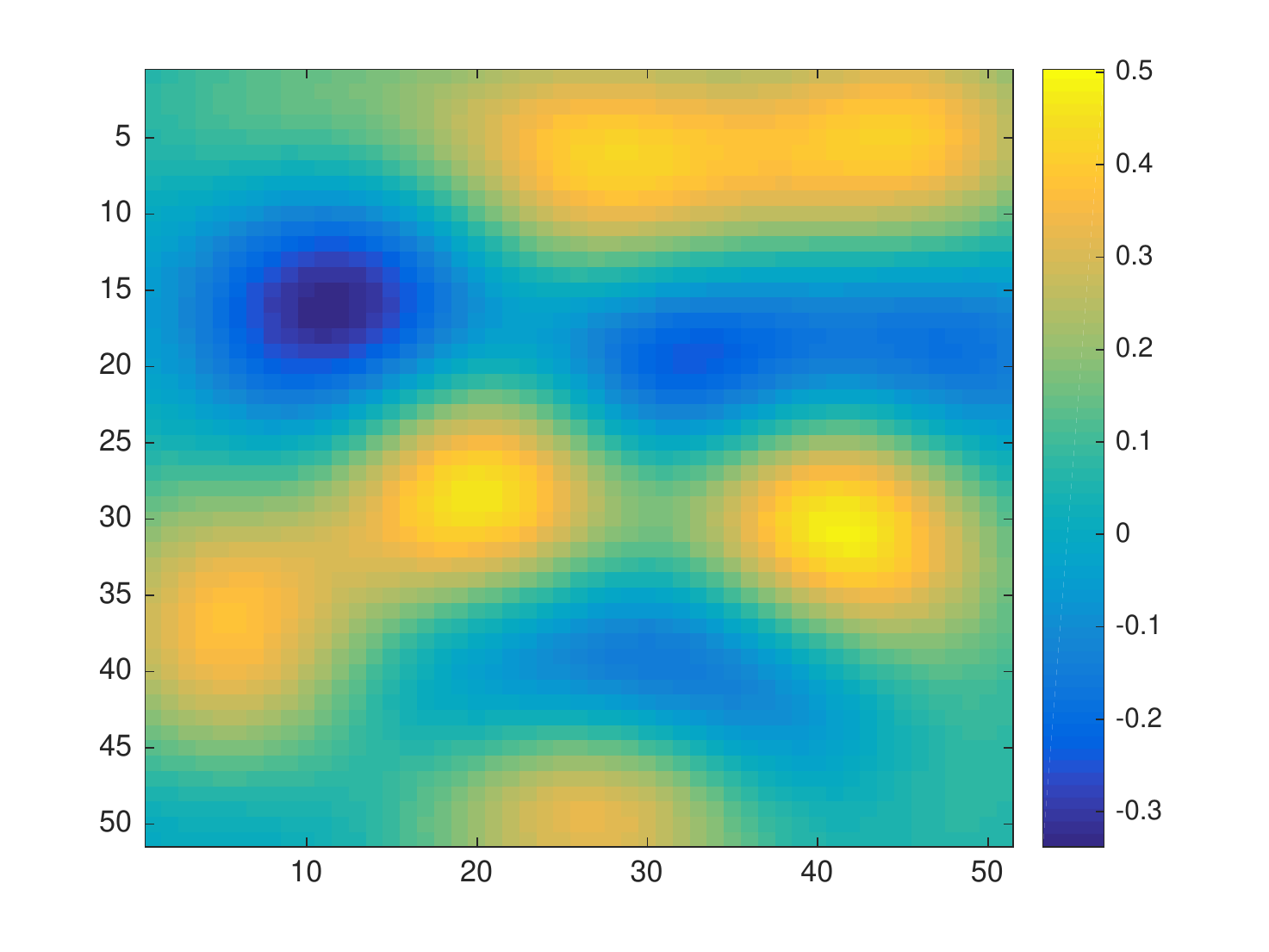}
        \caption{Mean log permeability field by VB method (complicated case).}
    \end{subfigure}\quad
    \begin{subfigure}[t]{0.4\textwidth}
        \centering
        \includegraphics[width=\textwidth]{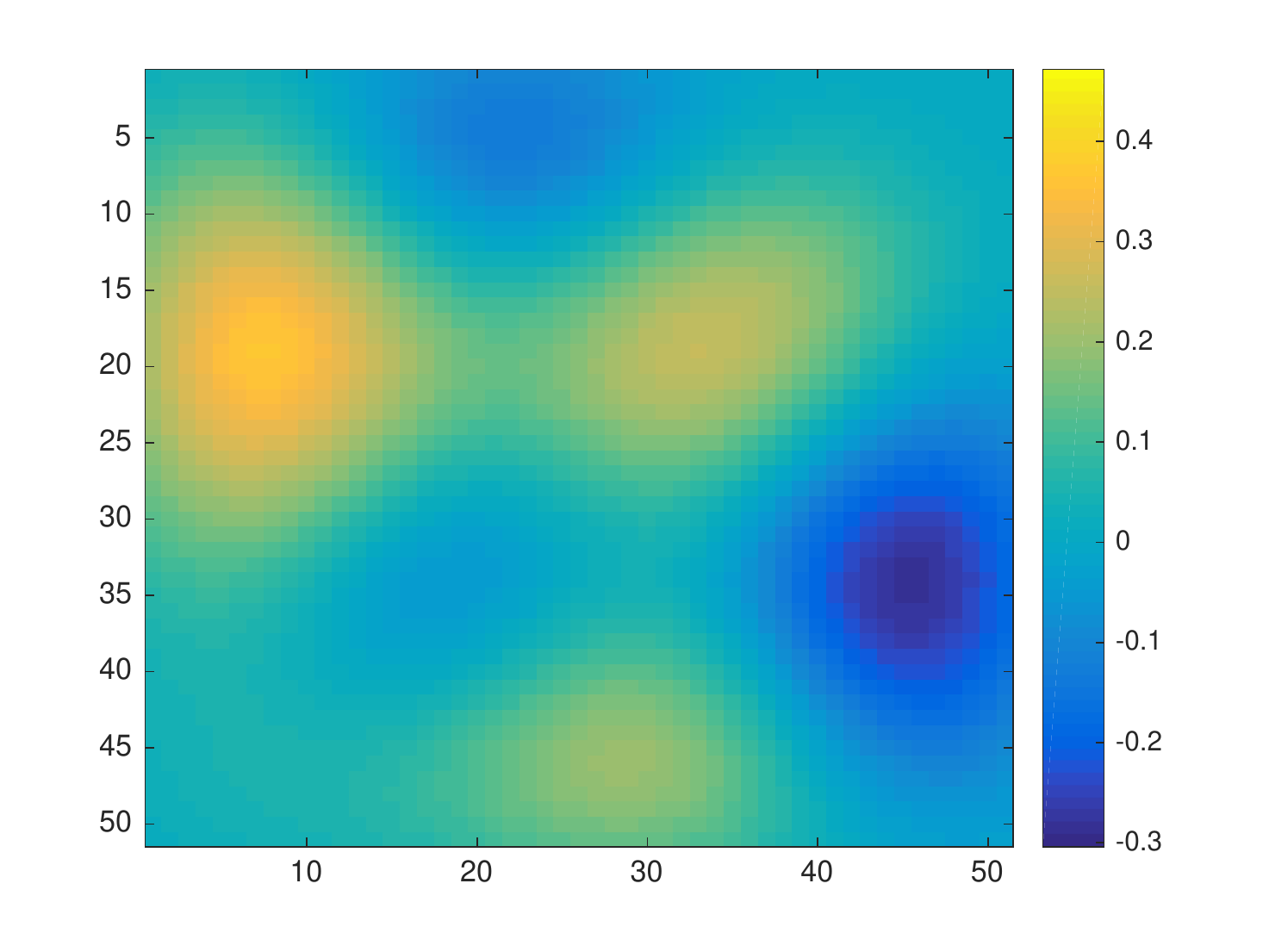}
        \caption{Mean log permeability field by MCMC method (smooth case).}
    \end{subfigure}\quad
    \begin{subfigure}[t]{0.4\textwidth}
        \centering
        \includegraphics[width=\textwidth]{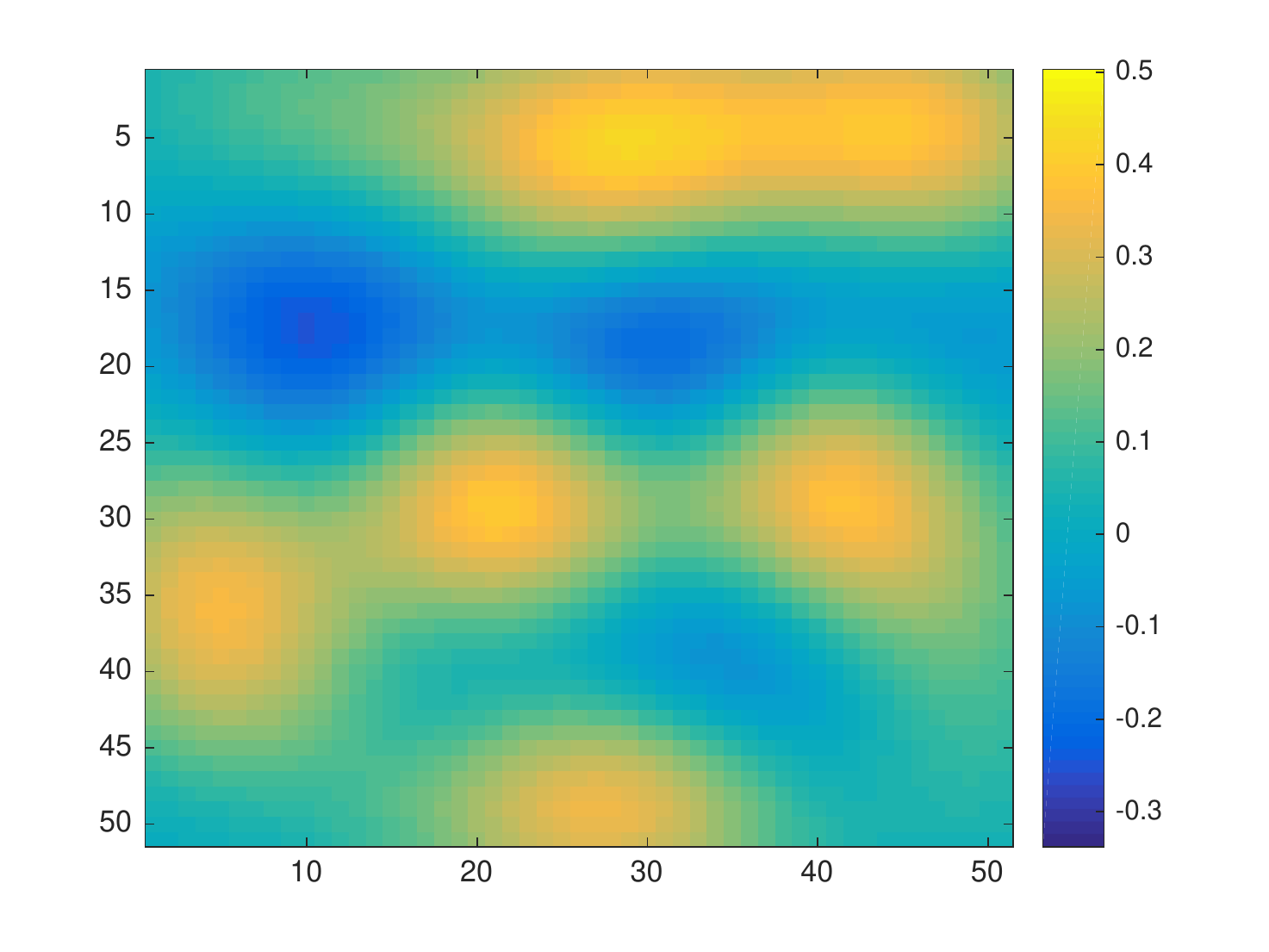}
        \caption{Mean log permeability field by MCMC method (complicated case).}
    \end{subfigure}
    \caption{Inverse problem result. From top to bottom: reference log permeability, Mean log permeability field by Variational Bayesian method and Mean log permeability field by MCMC method. From left to right: fewer observations (9 points) with smooth reference field to more observations (100 points) with complicate reference field.}
    \label{fig:ipresult}
\end{figure}

\begin{figure}
    \centering
    \begin{subfigure}[t]{0.4\textwidth}
        \centering
        \includegraphics[width=\textwidth]{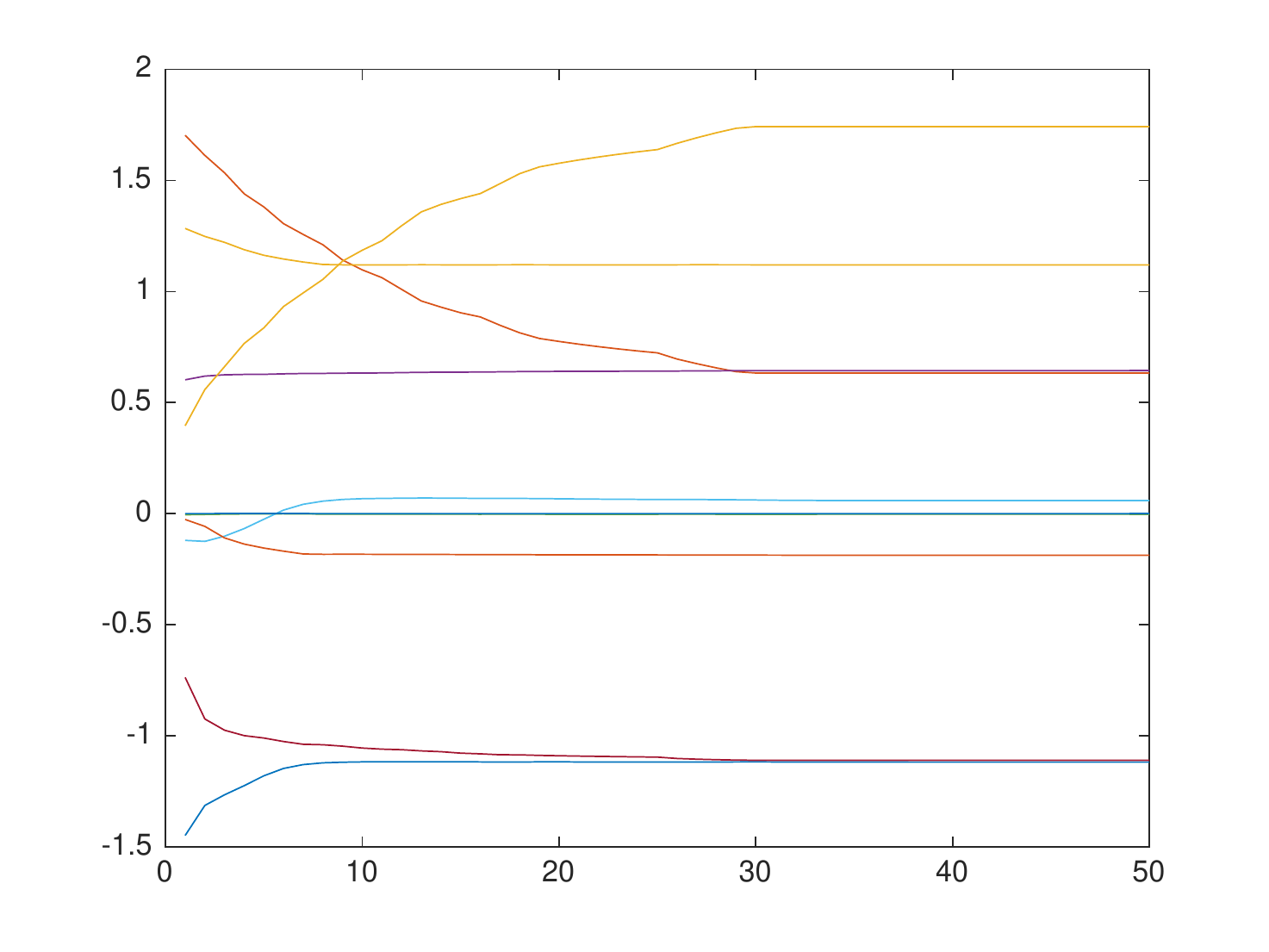}
        \caption{Trace of $E\theta_j$ in Variational Bayesian method (smooth case).}
    \end{subfigure}\quad
    \begin{subfigure}[t]{0.4\textwidth}
        \centering
        \includegraphics[width=\textwidth]{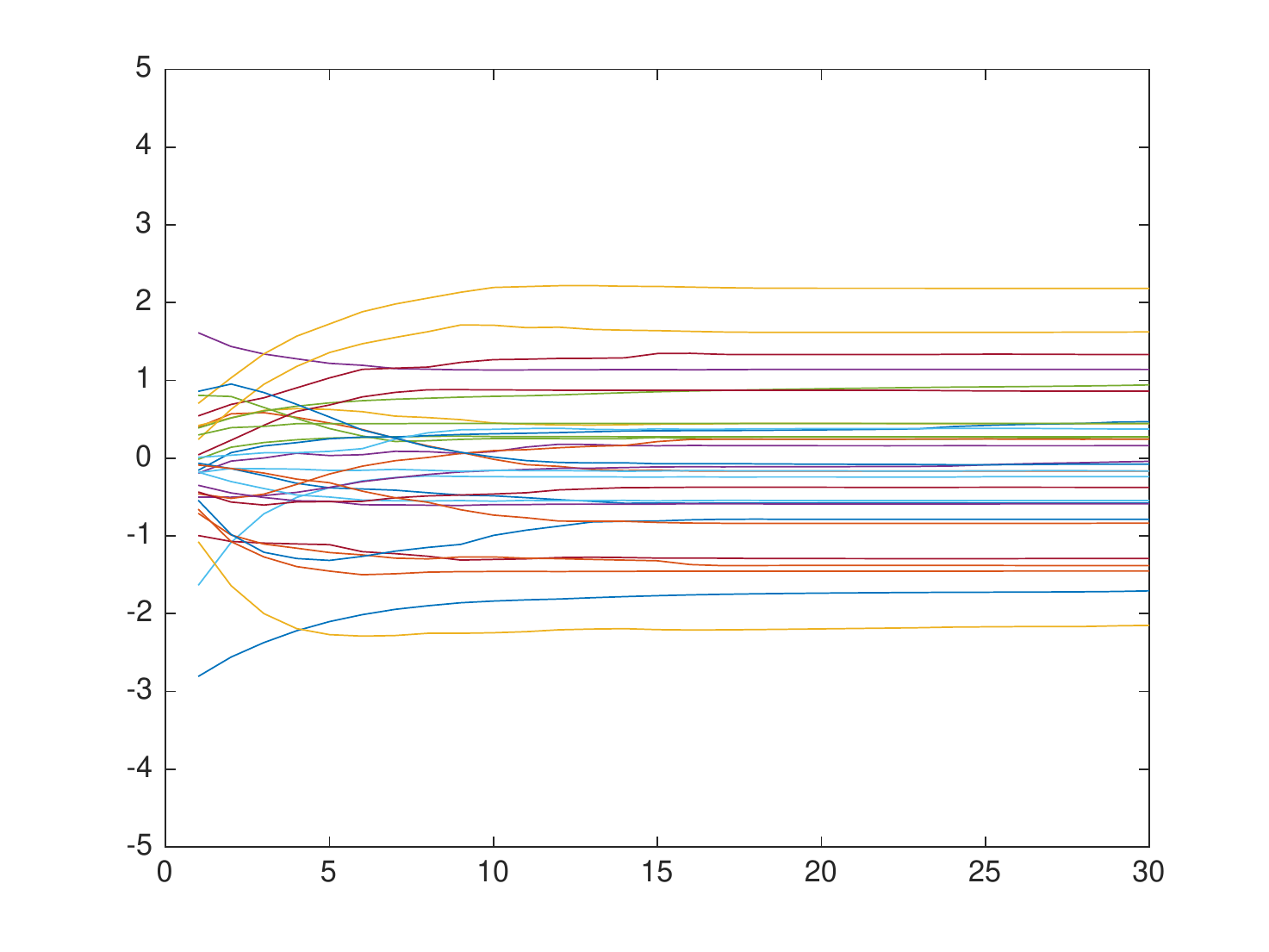}
        \caption{Trace of $E\theta_j$ in Variational Bayesian method (complicated case.}
    \end{subfigure}
    \quad
    \begin{subfigure}[t]{0.4\textwidth}
        \centering
        \includegraphics[width=\textwidth]{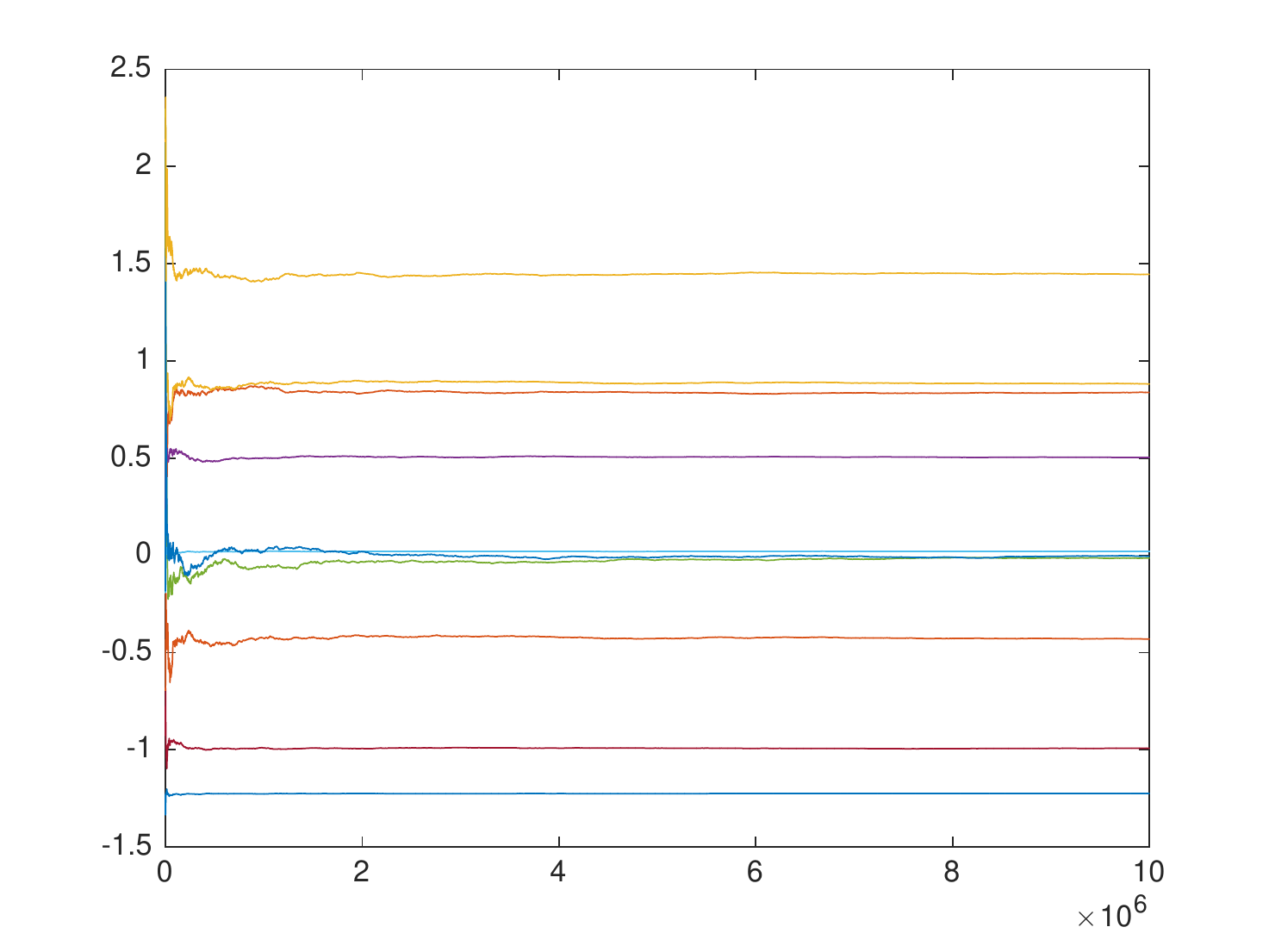}
        \caption{Trace of $E\theta_j$ in MCMC method (smooth case).}
    \end{subfigure}\quad
    \begin{subfigure}[t]{0.4\textwidth}
        \centering
        \includegraphics[width=\textwidth]{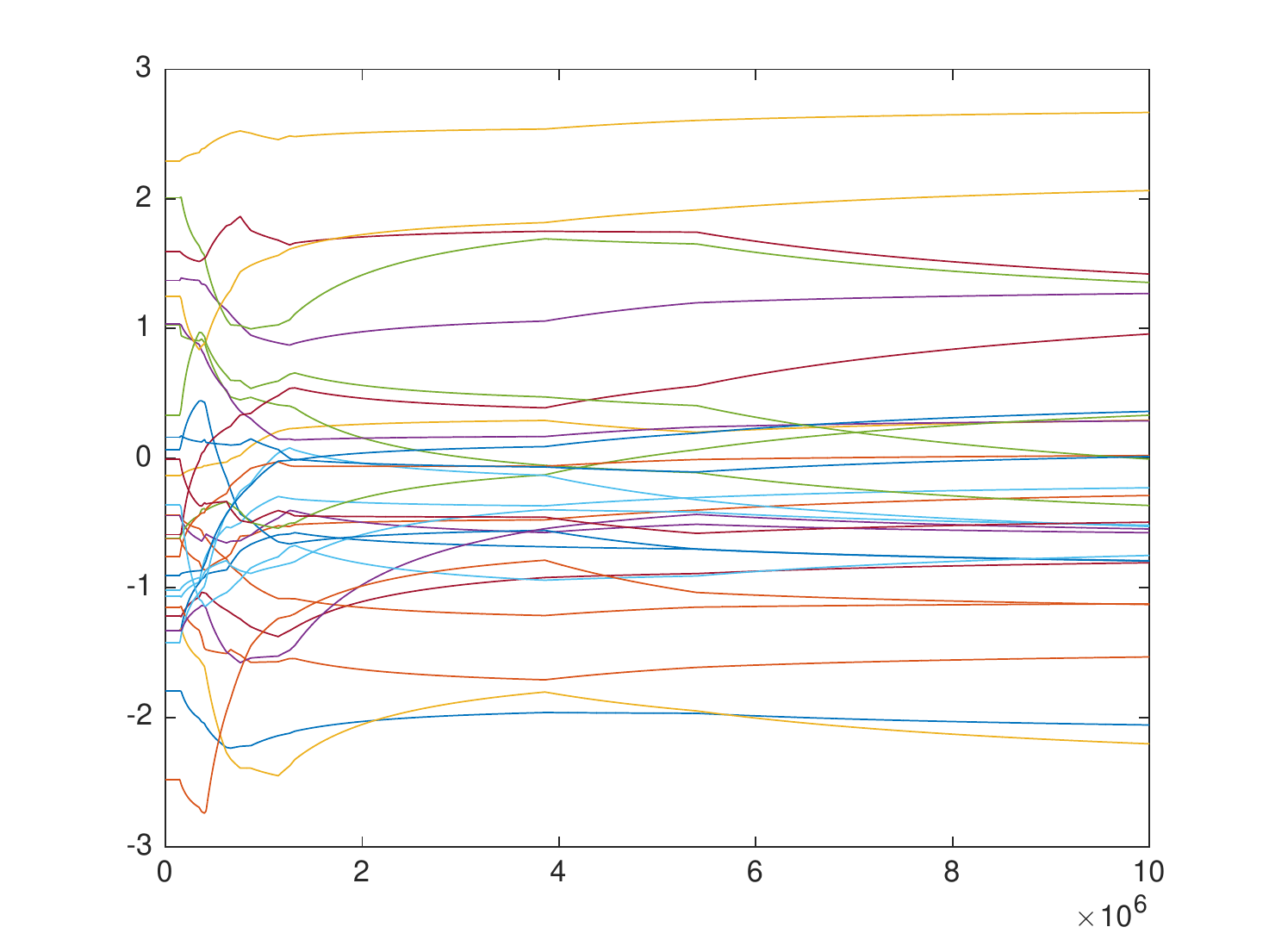}
        \caption{Trace of $E\theta_j$ in MCMC method (complicated case).}
    \end{subfigure}
    \caption{Trace of $E\theta_j$. Horizontal axis indicates the number of iterations, and vertical axis indicates the value of $E\theta_j$'s.}
    \label{fig:trace}
\end{figure}

\begin{figure}
    \centering
    \begin{subfigure}[t]{0.4\textwidth}
        \centering
        \includegraphics[width=\textwidth]{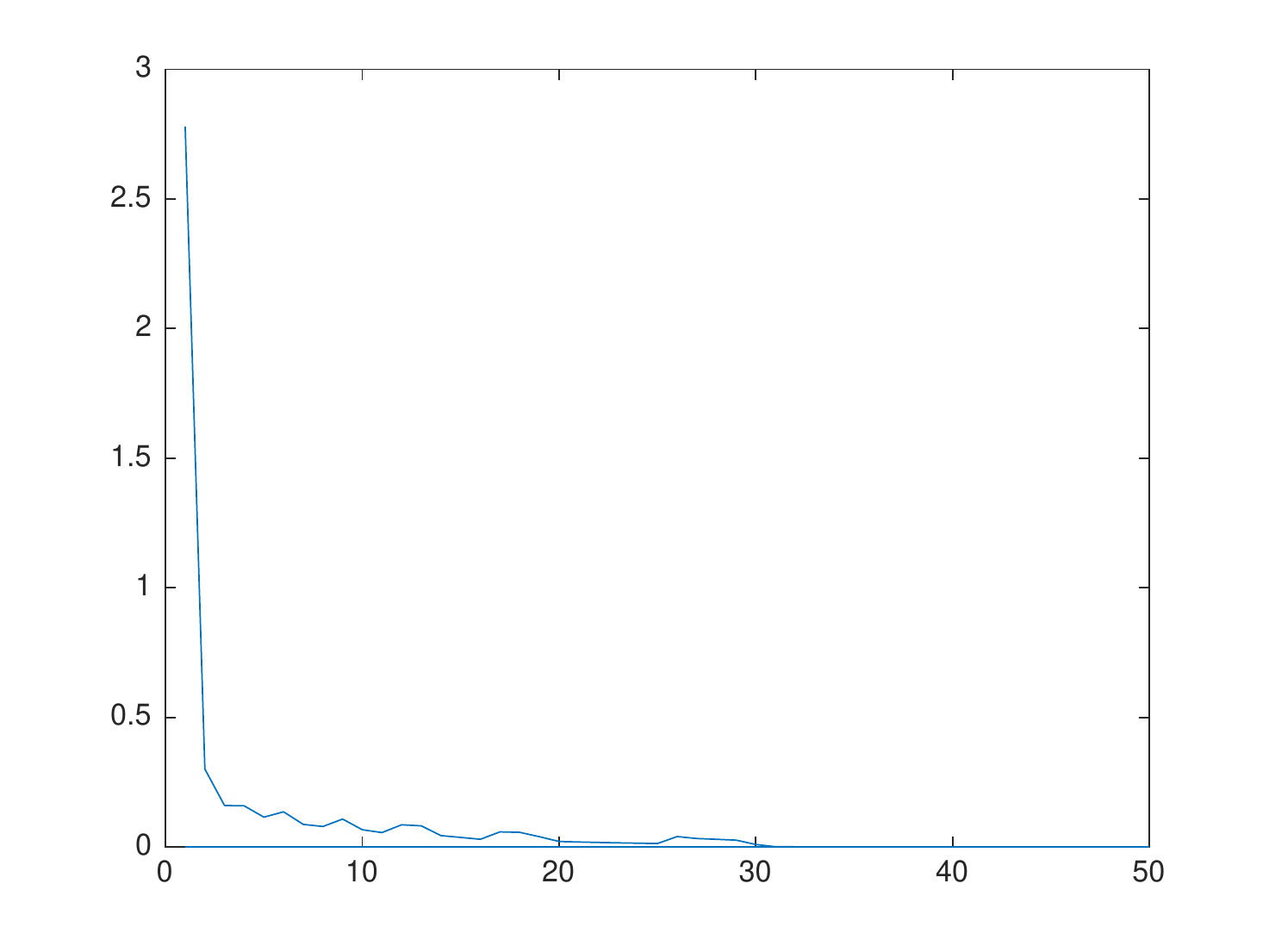}
        \caption{Plot of $\delta\mu$ (smooth case).}
    \end{subfigure}\quad
    \begin{subfigure}[t]{0.4\textwidth}
        \centering
        \includegraphics[width=\textwidth]{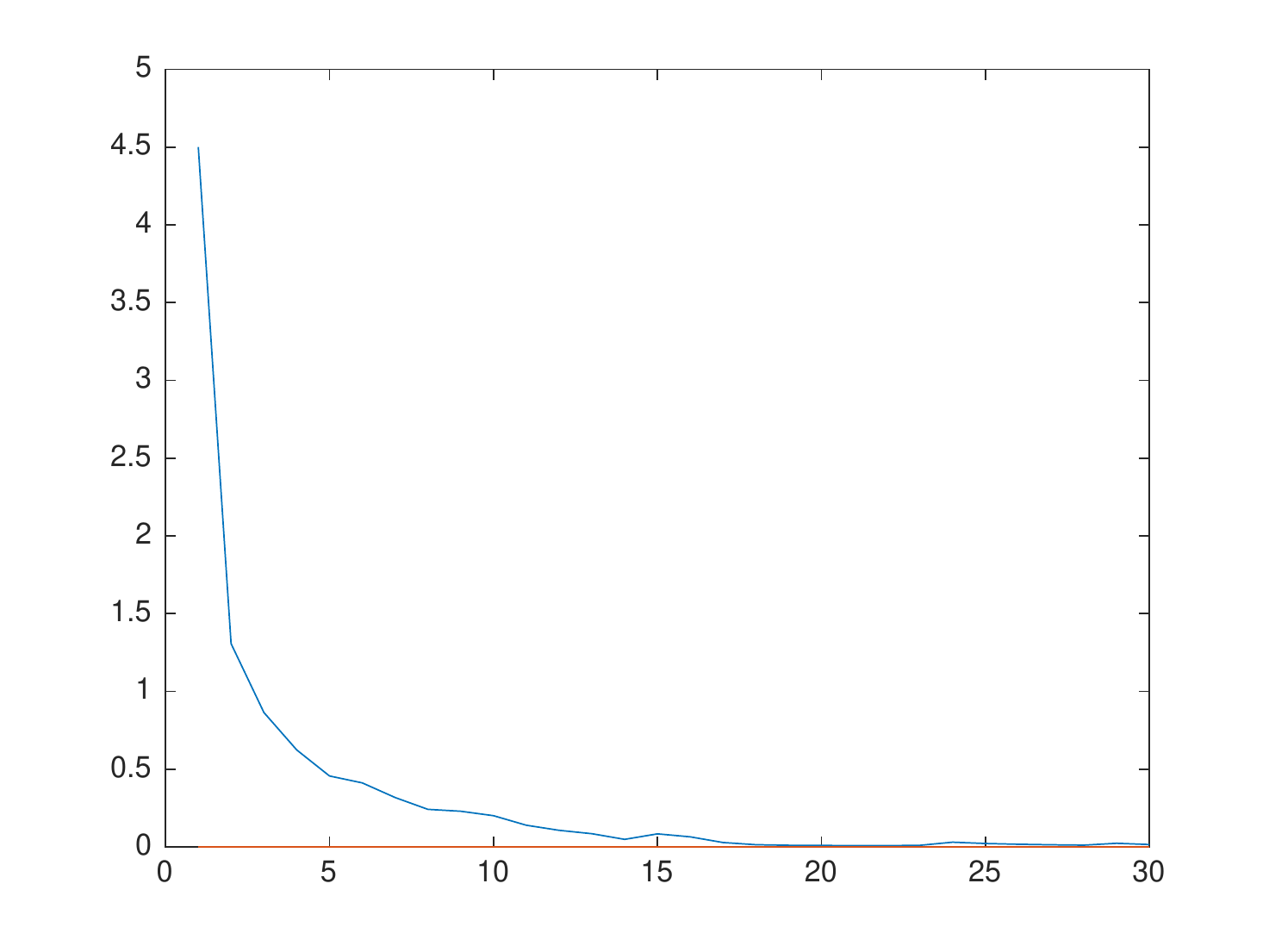}
        \caption{Plot of $\delta\mu$ (complicated case).}
    \end{subfigure}
    \caption{Convergence Plot of Variational Bayesian Method. Horizontal axis indicates the number of iterations, and vertical axis indicates the value of $\delta\mu$.}
    \label{fig:VBcon}
\end{figure}

\begin{figure}
  \captionsetup[subfigure]{labelformat=empty}
  \foreach \i in {1,...,8} {%
    \centering
    \begin{subfigure}{0.2\textwidth}
    \includegraphics[width=2cm,height=2cm]{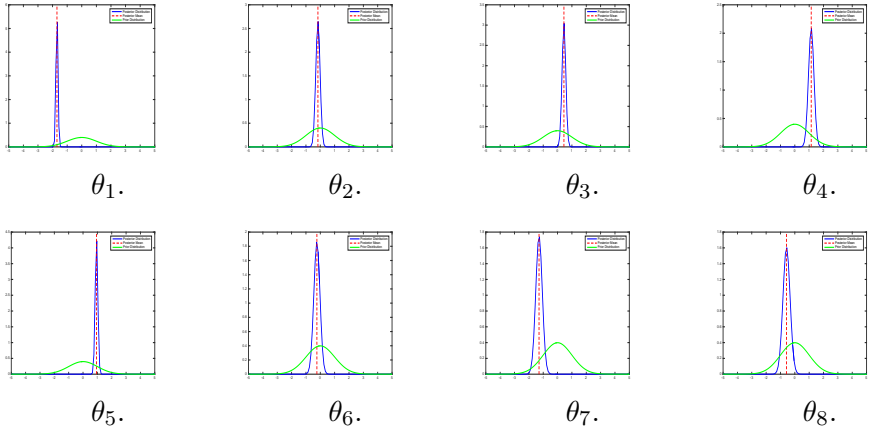}
    \caption{$\theta_{\i}$.}
    \end{subfigure}\quad
    \vspace{0.1cm}
  }
  \caption{Posterior and prior distribution of first 8 terms of $\theta_j$.}
  \label{fig:theta_distribution}
\end{figure}

\begin{table}
\centering
\begin{tabular}{|c|cc|}
  \hline
 Time(sec)  & smooth case & complicated case \\
\hline
  single iteration in VB method & 4.12 & 93.32 \\
  single iteration in MCMC method& 0.0001 & 0.0025 \\
  total time of VB till converge  & 123.6 & 1866 \\
  total time of MCMC ($10^7$ steps) & 978 & 24316 \\
\hline
\end{tabular}
\caption{Run time comparison between Variational Bayesian method and MCMC method.}
\label{rt1}
\end{table}

\FloatBarrier

\begin{appendices}
\section*{Appendix}

\section{Variational Approximation}

\subsection{Calculation for $\theta_k$}
The expectation and variance of $\theta_k$ namely $ E(\theta_k)$ and $ \sigma^2(\theta_k) $ are
\begin{align}
  E(\theta_k) &= \frac{\int_{\Omega_k} \theta_k q^*(\theta_k) \ d\theta_k}
  {\int_{\Omega_k}q^*(\theta_k) \ d\theta_k} \nonumber \\
  \sigma^2(\theta_k) &= \frac{\int_{\Omega_k} \left(\theta_k-E(\theta_k)\right)^2 q^*(\theta_k) \ d\theta_k}
  {\int_{\Omega_k}q^*(\theta_k) \ d\theta_k} \label{vtheta}
\end{align}
After we have the expectation and variance for all $\theta_k$, based on equation \eqref{KLE}, let $\phi_k = \sqrt{\lambda_k}\Phi_k$ we could get the expectation and variance of $\ln\kappa$ by
\begin{align}
  E(\ln\kappa) &= \sum_k E(\theta_k)\phi_k \nonumber \\
  Var(\ln\kappa) &= \sum_k \sigma^2(\theta_k)\phi_k^2 .\label{evkappa}
\end{align}

\subsection{Calculation for $a_{p,k}$}

From the variational formulation,

\begin{multline*}
  q(\theta_k,a_{p,k})= \frac{C}{\prod_i \sigma_{i,k}(\theta_k)}
        \cdot e^{-\frac{(\theta_k-\theta_0)^2}{2\sigma^2_0}-
        \frac{1}{2\sigma_y^2}\left(a_{p,k}^2E(R_p^2)v_p^\intercal v_p\right)}-c_p(\theta_k)
        \\
        \cdot \int_{R^{N_1-1}} e^{-A_{-p,k}^\intercal \Sigma A_{-p,k} + \beta^\intercal A_{-p,k} } \ dA_{-p,k}
\end{multline*}
and integrating we have

\begin{equation*}
   q^*(\theta_k,a_{p,k}) =
  \frac{1}{\sqrt{\det\Sigma} \cdot \prod_i \sigma_{i,k}(\theta_k)}
  e^{\left(\frac{1}{4}\beta^\intercal\Sigma^{-1}\beta-c_p(\theta_k)
  -\frac{(\theta_k-\theta_0)^2}{2\sigma^2_0}-
  \frac{1}{2\sigma_y^2}\left(a_{p,k}^2E(R_p^2)v_p^\intercal v_p\right)
  \right)}.
\end{equation*}

Then we could get the  marginal, expectation and variance of $a_{p,k}$ by numerically integrating over $\theta_k$.
\begin{align}
  E(a_{p,k}) &= \frac{\int_{R}\int_{\Omega_k} a_{p,k} q^*(\theta_k, a_{p,k})
  \ d\theta_k \ d a_{p,k}} {\int_{R}\int_{\Omega_k}q^*(\theta_k, a_{p,k}) \ d\theta_k \ d a_{p,k}}  \nonumber\\
  \sigma^2(a_{p,k}) &= \frac{\int_{R}\int_{\Omega_k} \left(a_{p,k}-E(a_{p,k})\right)^2 q^*(\theta_k, a_{p,k}) \ d\theta_k \ d a_{p,k}}
  {\int_{R}\int_{\Omega_k}q^*(\theta_k, a_{p,k}) \ d\theta_k \ d a_{p,k}} \label{vaij}.
\end{align}
We then introduce our calculation on $E(R_m),\ E(R_m^2)$ and $E(R_mR_n)$. Using our independence assumption between different $i$ and $j$, we have

\begin{align}
  E(R_m) &= E_{-k}(\prod_{j\neq k}a_{m,j}) &=& \prod_{j\neq k} E(a_{m,j}) \nonumber\\
  E(R_m^2) &= \ \ \prod_{j\neq k} E(a_{m,j}^2)
   &=& \prod_{j\neq k} \left[ E^2(a_{m,j}) + \sigma^2(a_{m,j})\right] \nonumber
   \\
  E(R_mR_n) &= \prod_{j\neq k} E(a_{m,j}a_{n,j}) &=&
        \prod_{j\neq k} \left[ E(a_{m,j})E(a_{n,j}) + \mathrm{Cov}(a_{m,j},a_{n,j}) \right]. \label{RMN}
\end{align}

\section{Proof of convergence}
First we give the proof for convergence in the forward model. Then we prove the posterior consistency result.
\subsection*{Forward Model}
We use the following proposition to prove the main result.
\begin{proposition}\label{prop2}
The solution $(a_n,v_n)$ of problem \eqref{eq:av} satisfies that for any functions $(a,v)\in L^2(\Omega_\mu)\times H^1_0(\Omega_x)$:
\begin{equation}\label{eq:EL}
  \int_\Omega \kappa\nabla(a_n\otimes v_n)\cdot\nabla(a_n\otimes v+a\otimes v_n)
  = \int_\Omega f_{n-1}(a_n\otimes v+a\otimes v_n)
\end{equation}
and
\begin{equation}\label{eq:ELun}
  \langle u_n,(a_n\otimes v+a\otimes v_n) \rangle = 0,
\end{equation}
as $u_n$ is defined by \eqref{eq:un}.
\end{proposition}
\subsection*{Proof of Theorem \ref{conv1}}
Because $(a_n,v_n)$ satisfies \eqref{eq:av}, then from Proposition \ref{prop2}
\begin{align*}
  \|u_{n-1}\|^2   &= \|u_n\|^2 + \| a_n\otimes v_n\|^2 \\
   &\geq \|u_n\|^2
\end{align*}
Thus $\|u_{n-1}\|$ converges and $\sum_n\int_\Omega \kappa |\nabla (a_n\otimes v_n)|^2<\infty$, which implies that $\|a_n\otimes v_n\| = \int_\Omega \kappa |\nabla (a_n\otimes v_n)|^2\rightarrow 0$ while $n\rightarrow \infty$. Furthermore,
{
\begin{align*}
  E_n &= \frac{1}{2}\int_\Omega\kappa|\nabla (a_n\otimes v_n)|^2-\int_{\Omega}f_{n-1}\ a_n\otimes v_n \\
   &= \frac{1}{2}\int_\Omega\kappa|\nabla (a_n\otimes v_n)|^2-\int_{\Omega}\kappa \nabla u_{n-1} \cdot \nabla (a_n\otimes v_n) \\
   &=  \frac{1}{2}\int_\Omega\kappa|\nabla (a_n\otimes v_n)|^2-\int_{\Omega}\kappa \nabla (u_n + a_n\otimes v_n) \cdot \nabla (a_n\otimes v_n)\\
   &=  -\langle u_n, \ a_n\otimes v_n \rangle -\frac{1}{2}\int_\Omega \kappa |\nabla (a_n\otimes v_n)|^2 \\
   &= -\frac{1}{2}\int_\Omega \kappa |\nabla (a_n\otimes v_n)|^2.
\end{align*}}
Therefore $\lim_{n\rightarrow\infty}E_n=0$.\\
Because $\|u_n\|$ is bounded, then up to the extraction of a subsequence we could assume that $u_n$ converges weakly to $u_{\infty}$ in $\Gamma$. Since $(a_n,v_n)$ is the minimizer of problem \eqref{eq:av}, for any $n$ and $(a,v)\in L^2(\Omega_\mu)\times H^1_0(\Omega_x)$,
{
\begin{equation*}
  \int_{\Omega}\frac{1}{2}\kappa|\nabla (a\otimes v)|^2-\int_{\Omega}\kappa \nabla u_n \cdot \nabla a\otimes v \geq E_n.
\end{equation*}
}
By taking $n\rightarrow\infty$ and combine $\lim_{n\rightarrow\infty}E_n=0$ we have
{
\begin{equation*}
  \int_{\Omega}\frac{1}{2}\kappa|\nabla (a\otimes v)|^2-\int_{\Omega}\kappa \nabla u_\infty \cdot \nabla a\otimes v \geq 0.
\end{equation*}
}

\subsection*{Proof of Proposition \ref{prop2}}

Use techniques similar to Proposition 2 from \cite{le2009results}, for any $(a,v)\in L^2(\Omega_\mu)\times H^1_0(\Omega_x)$ and $\forall \varepsilon\in \mathbb{R}$, we have
\begin{multline*}
  \int_{\Omega}\frac{1}{2}\kappa|\nabla (a_n+\varepsilon a)\otimes (v_n+\varepsilon v)|^2-\int_{\Omega}f_{n-1} (a_n+\varepsilon a)\otimes (v_n+\varepsilon v) \\
  \geq \int_{\Omega}\frac{1}{2}\kappa|\nabla (a_n\otimes v_n)|^2-\int_{\Omega}f_{n-1} (a_n\otimes v_n).
\end{multline*}

Then using similar argument as in  \cite{le2009results} the result follows.

Use techniques similar to Proposition 2 from \cite{le2009results}, for any $(a,v)\in L^2(\Omega_\mu)\times H^1_0(\Omega_x)$ and $\forall \varepsilon\in \mathbb{R}$, we have
\begin{multline*}
  \int_{\Omega}\frac{1}{2}\kappa|\nabla (a_n+\varepsilon a)\otimes (v_n+\varepsilon v)|^2-\int_{\Omega}f_{n-1} (a_n+\varepsilon a)\otimes (v_n+\varepsilon v) \\
  \geq \int_{\Omega}\frac{1}{2}\kappa|\nabla (a_n\otimes v_n)|^2-\int_{\Omega}f_{n-1} (a_n\otimes v_n).
\end{multline*}

Then using similar argument as in  \cite{le2009results} the result follows.

\section{Proof of Theorem \ref{consistency1}}
We will use a general  technique used in \cite{ghosh2003bayesian}
(Chapter 4.4) and, also in \cite{ghosal2000convergence}.
For our case, we start with defining the following quantities.

Let $Z_i=(y_i,x_i)$ be the observed value at point $x_i$ and $x_i \sim H$. Let $u^*=\sum_{i=1}^{N_1}\prod_{j=1}^{N_2} {a}_{i,j}^*v_i(x)$ be the true mean, where $a_{i,j}^*$ are the true values of the coefficients. Also, we assume $|v_i(x)|<C_0$ for $x \in \Omega$. We let $D(p,q)$ be the Kullback-Leibler {\it KL} and $H^2(p,q)$ be the Hellinger distance between two densities $p$ and $q$. Let $f^*(Z)=f^*(y|x)h(x)$ be the true data generating density, where $f^*(y|x)$ denotes the density corresponding to true mean $u^*(x,\kappa^*)$.

Let $\psi$ be the prior parameters $\{ a_{i,j}(\theta_j) \}_{i,j}$ and $\Pi(\psi)$ be the prior distribution.
Define
 $v_\epsilon=\{\psi:\int( \sqrt{f(Z)}-\sqrt {f^*(Z)})^2 dxdy<\epsilon \}$.
Then
\begin{eqnarray}
\Pi(v_\epsilon^c|data)&=&\Phi_M\Pi(v_\epsilon^c \cap \mathscr{K}|data)+(1-\Phi_M)\Pi(v_\epsilon^c \cap \mathscr{K}|data)+\Pi(v_\epsilon^c \cap \mathscr{K}^c|data) \nonumber\\
&\leq &\Phi_M\Pi(v_\epsilon^c \cap \mathscr{K}|data)+\frac {(1-\Phi_M)\int_{v_\epsilon^c \cap \mathscr{K}}\prod_{i=1}^M \frac{f(Z_i|\psi )}{f^*(Z_i|\psi )}\Pi(\psi) d\psi}{\int_{K_{\epsilon_2}}\prod_{i=1}^M \frac{f(Z_i|\psi )}{f^*(Z_i|\psi )}\Pi(\psi)d\psi} \nonumber \\
&  &+\frac {\int_{v_\epsilon^c \cap \mathscr{K}^c}\prod_{i=1}^M \frac{f(Z_i|\psi )}{f^*(Z_i|\psi )}\Pi(\psi)d\psi}{\int_{K_{\epsilon_2}}\prod_{i=1}^M\frac{f(Z_i|\psi )}{f^*(Z_i|\psi )}\Pi(\psi)d\psi} \nonumber \\
&=&\Phi_M\Pi(v_\epsilon^c \cap \mathscr{K}|data)+(1-\Phi_M)\frac{{\bf I_{1n}}}{{\bf I_{1d}}}+\frac{{\bf I_{2n}}}{{\bf I_{2d}}},
\label{post}
\end{eqnarray}
where $K_{\epsilon_2}$ is some $\epsilon_2$ {\it KL} neighborhood around $f^*$ and $\mathscr{K}$ is a compact set.
Here, $\mathscr{K}$ is a compact sieve, the prior probability of
$\mathscr{K}^c$
decreases exponentially with increasing $M$. Here, $\Phi_M$ is a test function, which we introduce later.

Next we derive the set $\mathscr{K}$ and cover $\mathscr{K}$ with 'relatively small' number of Hellinger balls and as the result holds for each of the balls, combining them gives us the proof. The last part is done by constructing exponentially powerful test statistics between two non-intersecting Hellinger balls (\cite{ghosh2003bayesian}-Chapter 4.4.1).

\noindent \underline{\it Derivation of the set $\mathscr{K}$}

We use a well known result for Gaussian process (GP). Given that $a_{i,j}( )$'s are supported in a compact subset of  the real line $ \mathscr{R}$, for each of the GP path, we have $P(\sup |a_{i,j}|>d_M) \leq e^{-\alpha d_M^2}; \alpha>0$; see \cite{tokdar2007posterior}. Choosing $d_M=\sqrt{M}$, we have $P(\sup_{i,j} |a_{i,j}|>d_M) \leq e^{-\beta M}$,
for some $\beta>0$. Hence, $\mathscr{K}=\{ |a_{i,j}|\leq d_M \}_{i,j}$ and $\Pi(\mathscr{K}^c) \leq e^{-\beta M}$.

\noindent \underline{\it Covering number  of  $\mathscr{K}$}

On $\mathscr{K}$, $|a_i(\mu)|\leq M^{.5N_2}$, we can have at most $\mathscr{N}(\epsilon',\mathscr{K})=$ $(\frac{C_0N_1 2M^{.5N_2}}{\beta_1\epsilon'})^{N_1}$ many grids of $a_{i,j}$'s, such that  for any point in  $\mathscr{K}$ with corresponding mean function value $u(x)$, we have a grid point  $\psi$ and corresponding value $\hat{u}$ such that $\sup |u-\hat{u}| <\beta_1\epsilon'$, $\beta_1>0$.  Hence,  choosing $\beta_1$ appropriately,  we have $D(f(u),f(\hat{u}))<\epsilon'^2$ and $H^2(f(u),f(\hat{u}))<\epsilon'^2$ (as {\it KL} distance dominates Hellinger distance).
Therefore, we have the log Hellinger covering number of $\mathscr{K}$, \[log(\mathscr{N}(\epsilon'^2,\mathscr{K},H))=o(M).\]

\noindent \underline{\it Sufficient prior mass around $f^*(Z)$}

Let, $A_{c,\delta}=\{\psi: sup_{i,j}|a_{i,j}-a_{i,j}^*|<c\sqrt{\delta}\}$ and then
Also,  $\Pi(A_{c,\delta})>0$ for any $c,\delta>0$. We can choose  $c$ small enough such that {\it KL} distance between $f^*(Z)$ and $f(Z,\psi)$ is less than $\delta$ for $\psi \in A_{c,\delta}$.   Hence $\Pi(K_\delta)>0$ for any $\delta$ {\it KL} neighborhood of $f^*(Z)$.

\noindent \underline{\it Combining all the parts}

Let, $4\epsilon_2<\text{min}\{\beta \epsilon,\epsilon\}$ and $\epsilon'^2=\epsilon_2>\delta$, then from \eqref{post}  as $M \rightarrow \infty$, $e^{M\epsilon_2}{\bf I_{2d}}>1$   and $e^{M\epsilon_2}{\bf I_{2n}}\rightarrow 0$ with probability 1 (see \cite{tokdar2007posterior}). Hence, $\frac{{\bf I_{2n}}}{{\bf I_{2d}}} \rightarrow 0$ with probability 1.

Let $H(\epsilon)$ be the Hellinger  ball around $f^*(Z)$ with distance $\sqrt{\epsilon}$; (i.e $H^2 \leq \epsilon$).
 For, ${\bf I_1}=\frac{{\bf I_{1n}}}{{\bf I_{1d}}}$, we can show there exists test  between $f^*$ and $H(\epsilon)^c \cap \mathscr{K}$ (\cite{tokdar2007posterior};\cite{ghosal2000convergence}--Theorem 2.1 proof, \cite{ghosal2000convergence}--Section 7: existence of tests)
 $\Phi_M$, such that \[E_{f^*}(\Phi_M)\leq e^{-\alpha_1M} \text{ and  }sup_{\mathscr{K}\cap H({\epsilon})^c}E_{f(Z,\psi)}(1-\Phi_M)\leq e^{-\alpha_1M}\]
 where $\alpha_1=.5\epsilon$. We use this $\Phi_M$ in equation \eqref{post}.

 From equation \eqref{post}, \[E_{f^*}(\Phi_M\Pi(v_\epsilon^c \cap \mathscr{K}|data))\leq E_{f^*}(\Phi_M)\leq e^{-\alpha_1M}.\]
 Then using Markov inequality and Borel-Cantelli lemma (\cite{ghosh2003bayesian}--chapter 4.4) $\Phi_M\Pi(v_\epsilon^c \cap \mathscr{K}|data)$ converges to zero almost surely.

 Also, $e^{M\epsilon_2}{\bf I_{1d}}>1$ with probability one for large $M$ and $E_f^*(1-\Phi_M) \frac{{\bf I_{1n}}}{{\bf I_{1d}}}\leq \text{sup}E_{f\in \mathscr{K}\cap H(\epsilon)^c}(1-\Phi_M)e^{M\epsilon_2}<e^{-M\epsilon_2}$. Hence, following the argument of Markov inequality and Borel-Cantelli lemma (see \cite{ghosh2003bayesian}),  $\frac{{\bf I_{1n}}}{{\bf I_{1d}}}\rightarrow 0$ with probability one, as $M\rightarrow \infty$.

 As a result, we have $\Pi(v_\epsilon|data) \rightarrow 1$, with probability 1 as
the number of observations $M$
goes to infinity. 

\end{appendices}
\section*{Acknowledgements}
We would like to thank the partial support from NSF 1620318, the U.S. Department of Energy Office of Science, Office of Advanced Scientific Computing Research,  Applied Mathematics Program under award number DE-FG02-13ER26165 and 
National Priorities Research Program grant NPRP grant 7-1482-1278 from the Qatar National Research Fund.
\section*{References}
\bibliography{bibfile}

\end{document}